\documentclass[onecolumn]{IEEEtran}
\usepackage{amsmath,amsfonts}

\usepackage{array}
\usepackage[caption=false,font=normalsize,labelfont=sf,textfont=sf]{subfig}
\usepackage{textcomp}
\usepackage{stfloats}
\usepackage{url}
\usepackage{verbatim}
\usepackage{graphicx}
\usepackage[noadjust]{cite}
\usepackage{graphicx}
\usepackage{amsthm}
\usepackage[table]{xcolor}
\usepackage{algorithm}
\usepackage{algpseudocode}
\usepackage{tikz}
\usepackage{amssymb}
\usepackage{mathtools}
\usepackage{overpic}
\usepackage{algorithm}
\usepackage{algpseudocode}

\DeclarePairedDelimiter\abs{\lvert}{\rvert}

\DeclarePairedDelimiter\floor{\lfloor}{\rfloor}
\DeclarePairedDelimiter\parenv{\lparen}{\rparen}

\DeclarePairedDelimiter\set{\{}{\}}

\theoremstyle{plain}
\newtheorem{theorem}{Theorem}
\newtheorem{lemma}[theorem]{Lemma}
\newtheorem{conjecture}[theorem]{Conjecture}

\newtheorem{corollary}[theorem]{Corollary}
\newtheorem{definition}[theorem]{Definition}

\newtheorem{remark}[theorem]{Remark}

\renewcommand{\leq}{\leqslant}

\renewcommand{\geq}{\geqslant}

\newcommand{\cA}{\mathcal{A}}

\newcommand{\cC}{\mathcal{C}}

\newcommand{\cH}{\mathcal{H}}

\newcommand{\cL}{\mathcal{L}}
\newcommand{\cM}{\mathcal{M}}
\newcommand{\cN}{\mathcal{N}}
\newcommand{\cO}{\mathcal{O}}
\newcommand{\cP}{\mathcal{P}}

\newcommand{\cS}{\mathcal{S}}

\newcommand{\cU}{\mathcal{U}}

\newcommand{\F}{\mathbb{F}}
\newcommand{\C}{\mathbb{C}}
\newcommand{\N}{\mathbb{N}}
\newcommand{\R}{\mathbb{R}}
\newcommand{\Z}{\mathbb{Z}}
\newcommand{\E}{\mathbb{E}}
\newcommand{\V}{\mathbb{V}}

\newcommand{\rcomment}[1]{\textcolor{red}{#1}}
\newcommand{\bcomment}[1]{\textcolor{blue}{#1}}

\newcommand{\eqdef}{\triangleq}

\newcommand{\rev}[1]{\overleftarrow{#1}}

\DeclareMathOperator{\supp}{supp}
\DeclareMathOperator{\Tr}{Tr}

\DeclareMathOperator{\lcm}{lcm}
\DeclareMathOperator{\ord}{ord}
\DeclareMathOperator{\BCH}{BCH}
\DeclareMathOperator{\Melas}{Melas}
\DeclareMathOperator{\wt}{wt}

\title{On the $b$-Symbol Weight Distribution of Dual Melas Codes and a Poisson Randomness Postulate}
\author{}
\date{June 2026}

\begin{document}

\author{
Gabriel~Sac~Himelfarb,~\IEEEmembership{Graduate Student Member,~IEEE} and Moshe~Schwartz,~\IEEEmembership{Fellow,~IEEE}%

\thanks{Gabriel Sac Himelfarb is with the Department of Electrical and Computer Engineering, McMaster University, Hamilton, ON, L8S 4K1, Canada (e-mail: sachimeg@mcmaster.ca).
}%
\thanks{Moshe Schwartz is with the Department of Electrical and Computer Engineering, McMaster University, Hamilton, ON, L8S 4K1, Canada (e-mail: schwartz.moshe@mcmaster.ca).}%
\thanks{This paper will be submitted in part to ISIT 2027.}%
\thanks{We acknowledge the support of the Natural Sciences and Engineering Research Council of Canada (NSERC), RGPIN-2026-06805 and CGRS D-611819-2026. Cette recherche a été financée par le Conseil de recherches en sciences naturelles et en génie du Canada (CRSNG), RGPIN-2026-06805 et CGRS D-611819-2026. This work was also supported in part by a grant of access to OpenAI models through the ChatGPT for Academic Researchers program. 
}
}

\maketitle

\begin{abstract}
    We prove a new lower bound on the $b$-symbol weight  of codewords in dual Melas codes, which is strictly better than known results. We employ a new strategy that translates a pattern counting problem in linear-feedback shift register sequences to a problem on orbits of binary polynomials under a certain action. We use the same techniques to provide an improved lower bound on the burst-covering radius of Melas codes and a new upper bound on the critical exponent of dual Melas codes. 
    
    We then define a new randomness postulate on (families of) binary sequences, based on the definition of infinite Poisson generic sequences. We conjecture that codewords in the dual Melas code satisfy this postulate, and we show partial progress in establishing this conjecture. We also show how a relaxed version of the problem can be completely resolved, and relate the conjecture to the first part of the paper, giving a justification of why the improved bounds on the $b$-symbol weight were somewhat to be expected. 
\end{abstract}


\section{Introduction}
\IEEEPARstart{M}{otivated} by applications to high-density data storage, Cassuto and Blaum introduced symbol-pair codes in \cite{CasBla10,CasBla11}. They considered a setting where readings from the channel are overlapping pairs of symbols rather than individual symbols. This model is adequate for scenarios where the read resolution is lower than the write resolution. Cassuto and Blaum's work prompted a fruitful new line of research, including the study of cyclic codes in the symbol-pair read channel \cite{CasLit11,YaaBruSie12}, and the construction of MDS symbol-pair codes \cite{LiGe17,DinGeZha18,LuoEzeLin22,MaLuo22}.

The symbol-pair channel generalizes naturally to $b$-symbol read channels, where $b$ consecutive symbols are observed together in each read, instead of just two. Despite there being numerous works in this area, such as \cite{YaaBruSie16,DinZhaGe18,ZhuShiOzb22}, the $b$-symbol weight distribution is only known for a few codes, e.g., \cite{ZhuShi24,ZhuXiaXiaHell25}. 

In \cite{ShiOzb21}, new bounds on the $b$-symbol weight for cyclic codes were given by a geometric approach. These may be viewed through the lens of pattern frequencies appearing in linear-feedback shift register sequences (LFSRs). Codewords in cyclic codes are LFSR sequences with connection polynomial equal to the generator polynomial of the dual code. Thus, bounds on the number of occurrences of an all-zero window in the LFSR yield bounds on the $b$-symbol weight. For instance, Niederreiter established bounds on LFSR pattern occurrences in \cite{Niede86}. It is true, however, that these bounds are not tight in some relevant cases, such as cyclic codes where the generator polynomial factors as the product of irreducibles of the same degree (see \cite[Section IV.A.]{SacSch26}). 

In this work we will focus in particular on the study of dual Melas codes, or equivalently, LFSR sequences that are obtained as the sum of a PN sequence (pseudo-noise sequence or $m$-sequence) and some shift of its reversal. We will show that if the length of the code is $2^k-1$, then its $k$-symbol weight is lower bounded by $2^k-1-2^{\frac{2k(1+o(1))}{\log_2 k}}$. This is strictly better than the known lower bounds of the order $2^k-\Theta(2^{k/2})$ which can be obtained through the application of the Weil-Carlitz-Uchiyama bound on character sums \cite{SacSch26} or the Hasse-Weil bound \cite{ShiOzb21}. The main strategy consists in translating the problem from the LFSR domain to studying orbits of binary polynomials under a certain action. We can then take advantage of the algebraic structure of the ring of polynomials, in particular, unique factorization, in order to obtain this improved result. We will see how this technique also allows us to give an improved lower bound for Melas codes on a recently introduced parameter, the burst-covering radius~\cite{SacSch26}, and a new upper bound on the critical exponent of dual Melas codes. The connection between LFSRs and polynomial orbits was previously used by the authors to derive upper bounds on the burst-covering radius of cyclic codes \cite{SacSch26}, albeit in the opposite direction, that is, the original characterization of the burst-covering radius in terms of degrees of polynomials in certain orbits was translated into the LFSR realm, where previously known results could be employed.

The second half of this paper presents a new randomness postulate on (families of) binary sequences inspired by the definition of Poisson generic infinite sequences \cite{AlvBecMer23}. Golomb famously introduced his three randomness postulates in \cite{Gol67}, and showed that PN sequences satisfy all three. Notably, PN sequences do not satisfy the Poisson randomness postulate, which prompts the question: can we give any explicit examples of such sequences? We conjecture that indeed, the codewords in the dual Melas codes satisfy this postulate, and show numerical evidence for this. This observation is totally new as far as we are aware. As first steps towards establishing this conjecture, we show the convergence of the first three factorial moments of  relevant random variables to those of the Poisson $\cP(1)$ distribution, when averaging over all codewords of $\Melas(k)^\perp$. Showing convergence for all moments seems like a problem that escapes all currently known techniques and a considerable challenge. We also show how the conjecture is related to the first half of the paper, as it provides an explanation of why a bound like $2^k-2^{k/\log_2 k}$ was indeed to be expected.

The paper is organized in the following way. In Section \ref{sec:prelim} we introduce necessary concepts and notations, including results on the connection between the study of LFSR sequences and the study of certain polynomial orbits. In Section \ref{sec:patternfreqandweight} we prove stronger lower bounds on the $b$-symbol weight distribution of dual Melas codes, as well as estimates on the number of codewords with full $b$-symbol weight. In Section \ref{sec:burstcoveringradius} we give an improved lower bound on the burst-covering radius of Melas codes, as well as a new upper bound on the critical exponent of dual Melas codes. We finish that section with an improved algorithm for computing the burst-covering radius of Melas codes. In Section \ref{sec:poisson} we introduce the definition and intuition behind a new Poisson randomness postulate for binary sequences, we provide numerical evidence that supports the conjecture that codewords in the dual Melas code satisfy this postulate, and prove convergence of the first three factorial moments of certain random variables to those of the Poisson $\cP(1)$ distribution. We also show the connection to Section \ref{sec:patternfreqandweight}, and finally, we prove that under a certain relaxation, it is possible to fully prove convergence in distribution to $\cP(1)$ of the relevant random variables.

\section{Preliminaries and notation}\label{sec:prelim}

This section is dedicated to presenting the definitions, basic notions, and results that are required for the main claims in the sections to follow. We begin with linear-feedback shift register sequences (LFSRs), then introduce an equivalence relation on binary polynomials which will prove to be useful, and establish a relation between both. We also recall concepts from coding theory, including the definitions of Melas and BCH codes, the notions of $b$-symbol weight,  burst-covering radius and the critical exponent of linear codes.

Throughout the paper, let $\F_q$ denote the finite field of size $q$, and $\F_q^*\eqdef \F_q\setminus\set{0}$. Let $\F_q[x]$ denote the set of polynomials over $\F_q$ in the unknown $x$. If $A$ is a set, we use $\# A$ to denote the size of $A$.

\subsection{Linear-feedback shift registers (LFSRs)}
The \emph{trace} map, $\Tr_{\F_{2^k}/\F_2}:\F_{2^k}\rightarrow \F_2$, is defined as
\[\Tr_{\F_{2^k}/\F_2}(\alpha)\eqdef\alpha+\alpha^{2}+\dots+\alpha^{2^{k-1}}.\]
If the domain and codomain are understood from context, we simply write $\Tr(\cdot)$. Given $\alpha\in \F_{2^k}$, we denote by $M_{\alpha}(x)\in \F_2[x]$ the minimal polynomial of $\alpha$ over $\F_2$.
\begin{definition}
Given a polynomial of degree $r$, $g(x)=x^r+\sum_{i=0}^{r-1}m_i x^i\in \F_2[x]$, we define an \emph{LFSR sequence} of order $r$ and connection polynomial $g$ as a sequence $(a_j)_{j\geq 0}$ which satisfies the linear recurrence
\[
a_j = \sum_{i=0}^{r-1}m_i a_{j-r+i}
\]
for all $j\geq r$. The elements $a_0,\dots,a_{r-1}$ are called the initial conditions (or initial elements) of the sequence. 
\end{definition}

The following theorem summarizes results from \cite{GorKlap12} and \cite{Gol67}.

\begin{theorem}\label{thm:LFSRtracecharacterization}
Let $g\in\F_2[x]$. Then the following two claims hold:
       \begin{itemize}
           \item [1)] If $g$ is irreducible over $\F_2$, and $\alpha\in \F_{2^r}$ is a root of $g$, there exists some $\beta\in \F_{2^r}$, determined by the initial conditions $a_0,\dots,a_{r-1}\in\F_2$, such that for all $j\geq 0$,
        \[
        a_j = \Tr_{\F_{2^r}/\F_2}(\beta \alpha^j).
        \]

        \item[2)]  If $g$ factors into distinct irreducible polynomials of the same degree $k$, $g=\prod_{i=1}^e g_i$,  then for all $j\geq 0$,
        \[
        a_j = \Tr_{\F_{2^{k}}/\F_2}\parenv*{\sum_{i=1}^e \gamma_i \alpha_i^j},
        \]
        where $\alpha_i\in \F_{2^{k}}$ is a root of $g_i$ and $\gamma_i\in \F_{2^k}$, for all $1\leq i\leq e$.        
       \end{itemize}
      
\end{theorem}

\subsection{Polynomial orbits}

We now introduce polynomial orbits that will facilitate our analysis of certain LFSRs of interest. For a vector $v=(v_0,v_1,\dots,v_{n-1})\in\F_2^n$, we define the reversal of $v$ as $\rev{v}\eqdef(v_{n-1},\dots,v_1,v_0)$. In a similar fashion, assume $f\in\F_2[x]$ of degree $r$, and write $f(x)=x^r+\sum_{i=0}^{r-1} f_i x^i$. Then the \emph{reciprocal} of $f$ is defined as the polynomial with coefficients in the reverse order of $f$, i.e.,
\[
\rev{f}(x) \eqdef 1 + \sum_{i=0}^{r-1} f_i x^{r-i} = x^{\deg f}f(x^{-1}).
\]
If $f=\rev{f}$ we say it is \emph{self-reciprocal}.

\begin{definition}
    Let $\alpha\in\F_{2^k}$ be primitive, $k\geq 3$. We define the equivalence relation $\sim$ on the set $\set{f\in\F_2[x]:\deg(f)<2k}$, as:
    \[
    f\sim g \Longleftrightarrow \exists j\in \Z  \text{ such that } f\equiv x^j g\pmod{M_{\alpha}M_{\alpha^{-1}}}.
    \]
    We denote the equivalence class of a polynomial $f$ as $\cO_f$ (as it is the orbit of $f$ under the action of multiplication by $x$ modulo $M_{\alpha}M_{\alpha^{-1}}$). We say an orbit $\cO_f$ is \emph{degenerate} if $\alpha$ or $\alpha^{-1}$ is a root of $f$, or equivalently, a root of all of its elements.
\end{definition}

We emphasize that whenever we use the $\sim$ equivalence relation, we assume that $\alpha\in\F_{2^k}$, $k\geq 3$, and $\alpha$ primitive, are understood implicitly from the context. We present a sequence of three lemmas that study properties of the equivalence relation $\sim$. First, we look at the size of the equivalence classes.

\begin{lemma}
\label{lem:orbitsize}
The equivalence relation $\sim$ has:
\begin{itemize}
\item
    One degenerate orbit of size $1$, $\cO_0=\set{0}$.
\item
    Two degenerate orbits of size $2^k-1$ each, containing multiples of $M_\alpha$ and $M_{\alpha^-1}$ respectively.
\item
    $2^k-1$ non-degenerate orbits of size $2^k-1$ each.
\end{itemize}
\end{lemma}
\begin{IEEEproof}
Trivially $\cO_0=\set{0}$, and it is the only orbit of size $1$. In all other cases, the size of each orbit is the minimal $j>0$ such that $f\equiv x^j f \pmod{M_{\alpha}M_{\alpha^{-1}}}$, that is, such that $f(x)\cdot(1+x^j)\equiv 0 \pmod {M_{\alpha}M_{\alpha^{-1}}}$. If the orbit is non-degenerate, we deduce that $M_{\alpha}M_{\alpha^{-1}}|x^j+1$, and if it is degenerate, then either $M_{\alpha}|x^j+1$ or $M_{\alpha^{-1}}|x^j+1$. In any case, since $\alpha$ and $\alpha^{-1}$ have order $2^k-1$, we deduce that the minimal $j$ is $2^k-1$, and thus $\#\cO_f=2^k-1$ for every $f$.
    
We further deduce that there are exactly two non-zero degenerate orbits: one contains all $2^k-1$ non-zero polynomials of degree $<2k$ which are multiples of $M_\alpha$, while the other contains the multiples of $M_{\alpha^{-1}}$.

Finally, the number of non-zero orbits is  $\frac{2^{2k}-1}{2^k-1}=2^k+1$, and the number of non-degenerate orbits is $2^k-1$.
\end{IEEEproof}

Next, we present a simple test to see whether two polynomials are equivalent in a non-degenerate orbit.

\begin{lemma}
\label{lemma:equiv}
Let $\alpha\in\F_{2^k}$ be primitive, $k\geq 3$. If $\cO_f$ is non-degenerate, then $f\sim g$ if and only if 
    \[
f(\alpha)f(\alpha^{-1})=g(\alpha)g(\alpha^{-1}).
    \]
\end{lemma}
\begin{IEEEproof}
If $f\sim g$, by definition there exists $j$ such that $f(\alpha)=\alpha^j g(\alpha)$ and $f(\alpha^{-1})=\alpha^{-j} g(\alpha^{-1})$. Multiplying the two equations we get $f(\alpha)f(\alpha^{-1})=g(\alpha)g(\alpha^{-1})$.

For the converse, since $\cO_f$ is non-degenerate, $0\neq f(\alpha)f(\alpha^{-1})=g(\alpha)g(\alpha^{-1})$. Note that this also implies $\cO_g$ is non-degenerate. Since $\alpha$ is primitive, we have for some $j$
\[
\frac{f(\alpha)}{g(\alpha)}=\frac{g(\alpha^{-1})}{f(\alpha^{-1})} = \alpha^j,
\]
from which we extract
\begin{align*}
f(\alpha)-\alpha^j g(\alpha) & = 0, &
f(\alpha^{-1})-\alpha^{-j} g(\alpha^{-1}) &= 0.
\end{align*}
Thus, $M_\alpha M_{\alpha^{-1}} | f- x^j g$.
\end{IEEEproof}



\begin{lemma}\label{lemma:orbitequivdef}
Let $\alpha\in\F_{2^k}$ be primitive, $k\geq 3$, and $f,g\in\F_2[x]$. Write $f=x^a P$, $g=x^b Q$, such that $P,Q\in\F_2[x]$ and $x\nmid P,Q$. Then the following claims hold:
    \begin{itemize}
       \item [1)] If $\deg(P),\deg(Q)<k$, then
    \[
    f(\alpha)f(\alpha^{-1})=g(\alpha)g(\alpha^{-1}) \Longleftrightarrow P(x)\rev{P}(x)=Q(x)\rev{Q}(x),
    \]
    and $\deg(P)=\deg(Q)$.
    \item[2)] If $\deg (Q)\leq \deg(P) \leq k+s$ for some $s\geq 0$, then $f(\alpha)f(\alpha^{-1})=g(\alpha)g(\alpha^{-1})$ if and only if $P$ and $Q$ satisfy one of the following equations:
    \[
    P(x)\rev{P}(x)=x^{\deg (P)-\deg (Q)} Q(x)\rev{Q}(x)+ M_\alpha(x)M_{\alpha^{-1}}(x)\cdot R(x),
    \]
    where $R(x)\in\F_2[x]$ has degree $\leq 2s$ and $x^{2(\deg P-k)}R(x^{-1})=R(x)$. Notice that there are at most $2^{s+1}$ possible such choices for $R$.
    
    \item [3)] If $\deg(f),\deg(g)<k$ then
    \[
    f(\alpha)g(\alpha^{-1})=f(\alpha^{-1})g(\alpha) \Longleftrightarrow f(x)g(x^{-1})=f(x^{-1})g(x).
    \]
    \end{itemize}
\end{lemma}

\begin{IEEEproof}
%
    \textbf{Claim 1):} Suppose without loss of generality that $\deg (P)\geq \deg (Q)$. Consider the polynomial
\begin{equation}
\label{eq:defh}
h(x)\eqdef x^{\deg (P)}(P(x)P(x^{-1})+Q(x)Q(x^{-1})).
\end{equation}
If $f(\alpha)f(\alpha^{-1})=g(\alpha)g(\alpha^{-1})$ then $h(\alpha)=h(\alpha^{-1})=0$. Thus, $M_\alpha M_{\alpha^{-1}}|h$. Notice that $\deg (h) \leq 2\deg (P)<2k$, so $h=0$. This implies that $P(x)P(x^{-1})=Q(x)Q(x^{-1})$, hence $\deg (P)=\deg (Q)$. Thus, $h(x)=P(x)\rev{P}(x)+Q(x)\rev{Q}(x)=0$. 
    
On the other hand, if $P(x)\rev{P}(x)=Q(x)\rev{Q}(x)$, then necessarily $\deg (P)=\deg (Q)$, and so $P(\alpha)\rev{P}(\alpha)= Q(\alpha)\rev{Q}(\alpha)$ implies $f(\alpha)f(\alpha^{-1})=g(\alpha)g(\alpha^{-1})$.

\textbf{Claim 2):} Let $h$ be as in~\eqref{eq:defh}. If $f(\alpha)f(\alpha^{-1})=g(\alpha)g(\alpha^{-1})$ we deduce in the same way that $M_\alpha M_{\alpha^{-1}}|h$. Since $\deg (h)\leq 2\deg (P) \leq 2(k+s)$, we have $h = M_\alpha M_{\alpha^{-1}}\cdot R(x)$ for some $R\in\F_2[x]$ with $\deg (R) \leq 2s$. Then 
\[P\rev{P} = x^{\deg (P)-\deg (Q)}Q\rev{Q}+M_\alpha M_{\alpha^{-1}}R.\] 

Evaluating the above equation in $x^{-1}$ and using the fact that $M_\alpha M_{\alpha^{-1}}$, $P\rev{P}$ and $Q\rev{Q}$ are all self-reciprocal, we obtain the claimed identity on $R$.

For the converse, if $P$ and $Q$ satisfy one of the equations for some $R$, evaluating at $\alpha$ yields $P(\alpha)P(\alpha^{-1})=Q(\alpha)Q(\alpha^{-1})$, which implies $f(\alpha)f(\alpha^{-1})=g(\alpha)g(\alpha^{-1})$.

\textbf{Claim 3):} Similarly to 1), suppose without loss of generality that $\deg (f)\geq \deg (g)$, and consider 
\[h(x)\eqdef x^{\deg (f)}(f(x)g(x^{-1})+f(x^{-1})g(x)).\] 
The rest of the argument is analogous.
\end{IEEEproof}

We will also require the following technical counting lemma, whose proof can be found in the appendix.

\begin{lemma}\label{lemma:thirdmomentcounting}
The following holds:
    \begin{itemize}
        \item [1)]\[
    \#\set*{(f,g)\in (\F_2[x]-\set{0})^2 : \deg(f),\deg(g)<k, \frac{g(\alpha)}{g(\alpha^{-1})}=\frac{f(\alpha)}{f(\alpha^{-1})}} = \begin{cases}
        (3k-3)2^{k-2}+1 & k \text{ odd,}\\
        (3k-4)2^{k-2}+1 & k \text{ even.}
    \end{cases}
    \]

    \item[2)]\[
    \#\set*{(f,g)\in (\F_2[x]-\set{0})^2 : \deg(f),\deg(g)<k, g(\alpha)g(\alpha^{-1})=f(\alpha)f(\alpha^{-1})} = (k+1)2^{k-1}-1.
    \]
    \end{itemize}
    
\end{lemma}

\subsection{Connection between LFSRs and polynomial orbits}

We now establish a connection between the polynomial orbits defined above, and LFSR sequences. To that end, we need to introduce Galois-mode LFSRs.

\begin{definition}
Given a polynomial of degree $r$, $g(x)=\sum_{i=0}^{r-1}m_i x^i +x^r\in \F_2[x]$, we define the \emph{Galois-mode LFSR} of order $r$ and connection polynomial $g$, as a sequence generator with states of the form $(f_0,f_1,\dots,f_{r-1})\in \F_2^r$, and whose state transition is given by:
\[
(f_0,f_1, \dots, f_{r-1}) \longrightarrow (0,f_0,\dots,f_{r-2}) + f_{r-1}\cdot(m_0,\dots,m_{r-1}).
\]
The output of the LFSR is the sequence of elements $f_{r-1}$ from each state.
\end{definition}

\begin{theorem}[\cite{SacSch26,GorKlap12}]\label{thm:galois}
Given a Galois-mode LFSR of order $r$ with a connection polynomial $g\in\F_2[x]$,
\begin{itemize}
\item[1)] If we identify a state $(f_0,\dots,f_{r-1})$ with the polynomial $f(x)=\sum_{i=0}^{r-1} f_i x^i$, then the sequence of states corresponds to the sequence of polynomials $(x^jf \pmod g)_{j\geq 0}$, where $f$ is the initial state (initial load).

        \item [2)] The output $(a_j)_{j\geq 0}$ satisfies the linear recurrence
        \[
        a_j = m_{r-1}a_{j-1}+m_{r-2}a_{j-2}+\dots+m_0a_{j-r},
        \]
        for every $j\geq r$. 
    \end{itemize}
\end{theorem}

\begin{remark}\label{rem:LFSRpolycorrespondence}
In the setting of Theorem~\ref{thm:galois}, define the companion matrix 
    \[
    A \eqdef \begin{pmatrix}
        0 & 0 & \dots & 0 & m_0 \\
        1 & 0 & \dots & 0 & m_1 \\
        \vdots &  \ddots & \ddots & \vdots & \vdots \\
        0 & 0 & \ddots & 0 & m_{r-2}\\
        0 & 0 & \dots & 1 & m_{r-1}
    \end{pmatrix}.
    \]
    Let $f(x)=\sum_{i=0}^{r-1}f_i x^i$ be the initial load of the Galois-mode LFSR. The states are given by $A^0f, A^1f, A^2f, \dots$. Denote by $e$ the row vector $e=(0,0,\dots,0,1)$. Then the initial elements of the output LFSR sequence are:
    \[
    \begin{cases}
        e A^0 f & = a_0, \\
        e A^1 f & = a_1,\\
         & \vdots 
        \\
        e A^{r-1} f & = a_{r-1}.
    \end{cases}
    \]
    We claim that the vectors $e A^0, e A^1, \dots, e A^{r-1}$ are linearly independent. Indeed, it is not hard to verify inductively that $e A^i$ has zeros in coordinates $0,1,\dots,r-2-i$ and a $1$ in coordinate $r-1-i$ (as well as other possibly non-zero entries in positions $r-i,\dots,r-1$). 

    We deduce that for every choice of initial conditions $(a_0,\dots,a_{r-1})$, there exists a unique initial load $f$ that generates those initial elements for the output sequence in the Galois-mode LFSR. Thus, there is a one-to-one correspondence between sequences $(x^j f\pmod g)_{j\geq 0}$, and LFSR sequences with connection polynomial $g$. 
    
    Instantiating this for $g=M_\alpha M_{\alpha^{-1}}$, we see that there is a one-to-one correspondence between polynomial orbits $\cO_f$ and  LFSR sequences with connection polynomial $M_\alpha M_{\alpha^{-1}}$ (distinct up to shifting).
  \end{remark}  


\begin{theorem}[\cite{SacSch26}]\label{thm:galoisconsecutivezeros}
    Given $g\in \F_2[x]$ of degree $r$, consider the Galois-mode LFSR with connection polynomial $g$ and initial load $f\neq 0$. Denote by $(a_j)_{j\geq 0}$ the output sequence. Then
    \[
    \deg(x^jf \pmod{g}) = r-1-\max \set*{i: a_j=0, a_{j+1}=0, \dots, a_{j+i-1}=0}, 
    \]
    where the maximum is taken to be $0$ if $a_j=1$.
\end{theorem}

Let us instantiate the previous result in the particular case which will be useful throughout this work.

\begin{theorem}\label{thm:characterizationzerosdeg}
Let $g\in\F_2[x]$, $\deg(g)=2k$, and consider the Galois-mode LFSR with connection polynomial $g$ and initial load $f\neq 0$. Denote by $(a_j)_{j\geq 0}$ the output sequence. This sequence has a run of (at least) $k$ zeros at position $j$, $a_j=\dots = a_{j+k-1}=0$, if and only if $\deg (x^jf \pmod{g})<k$. More generally, it has a run of at least $t$ zeros at position $j$ if and only if $\deg (x^jf \pmod{g})<2k-t$. 
    \end{theorem}

\begin{IEEEproof}
This is a direct application of Theorem~\ref{thm:galoisconsecutivezeros}.
\end{IEEEproof}

Theorem \ref{thm:characterizationzerosdeg} establishes a connection between the degrees of polynomials in orbits $x^jf\pmod{g}$ and the occurrence of the window  $0\dots 0=0^k$ ($k$ consecutive zeros) in an LFSR sequence. In \cite{SacSch26}, this fact was used in one direction: known results for LFSR sequences were used to prove results on the degrees of the polynomials, which in turn yielded upper bounds on the burst-covering radius of cyclic codes (see the next subsection for the definition of the burst-covering radius). In the present work we will take advantage of this connection in the opposite direction: we will obtain new results on LFSR sequences by studying the corresponding polynomial orbits. 

Theorem \ref{thm:characterizationzerosdeg} admits a generalization that characterizes the occurrence of other window patterns besides $0^k$.

\begin{theorem}\label{thm:characterizationotherpatterns}
Let $g \in\F_2[x]$, $\deg(g)=2k$, and consider the Galois-mode LFSR with connection polynomial $g$ and initial load $f\neq 0$. Denote by $(a_j)_{j\geq 0}$ the output sequence. For each $w\in \F_2^k$, there exists a polynomial $h_w\in\F_2[x]$ of degree $<2k$, such that
    \[
    (a_j,\dots,a_{j+k-1})=w \Longleftrightarrow x^j f \pmod{g} \in\langle1,x,\dots,x^{k-1}\rangle+ h_w,
    \]
    that is, if and only if $x^j f\pmod g$ belongs to a certain affine subspace of polynomials of dimension $k$. 
\end{theorem}

\begin{IEEEproof}
    Define the companion matrix $A$ as in Remark \ref{rem:LFSRpolycorrespondence}. The states of the Galois-mode LFSR are $A^0f, A^1f, A^2f, \dots$. As before, denote $e=(0,0,\dots,0,1)$. Let $f'=x^j f \pmod g=\sum_{i=0}^{2k-1}f_i' x^i$. Thus, $(a_j,\dots,a_{j+k-1})=w$ if and only if
    \[
    \begin{cases}
        e A^0 f' & = w_0, \\
        e A^1 f' & = w_1,\\
         & \vdots 
        \\
        e A^{k-1} f' & = w_{k-1}.
    \end{cases}
    \]
    The vectors $e A^0, e A^1, \dots, e A^{k-1}$ are linearly independent by the same argument as in Remark \ref{rem:LFSRpolycorrespondence}. Thus, the above system consists of $k$ linearly independent equations on $2k$ unknowns $f'_0,\dots,f'_{2k-1}$. The solution is an affine space of dimension $k$. By Theorem \ref{thm:characterizationzerosdeg}, the solution to the homogeneous system is precisely $\langle 1,x,\dots,x^{k-1}\rangle$.
\end{IEEEproof}

We make two comments regarding Theorem~\ref{thm:characterizationotherpatterns}. First, we note that $h_{0^k}=0$ by Theorem \ref{thm:characterizationzerosdeg}. Second, by choosing an adequate representative, we may take $h_w$ to be a multiple of $x^k$. 

\subsection{Melas codes and related concepts}
An $[n,n-r]$ binary linear code, $\cC$, is an $(n-r)$-dimensional subspace, $\cC\subseteq\F_2^n$. We can specify the code $\cC$ through a parity-check matrix $H\in\F_2^{r\times n}$, such that $c\in \cC$ if and only if $Hc=0$, i.e., $\cC=\ker(H)$. Vectors of the form $Hv$, $v\in\F_2^n$, are called syndromes. The dual code of $\cC$, denoted $\cC^\perp$, is the linear code spanned by the rows of $H$.

Given a vector $v=(v_0,\dots,v_{n-1})\in\F_2^n$, a cyclic shift of $v$ is the vector $(v_{n-1},v_0,v_1,\dots,v_{n-2})$. A linear code $\cC$ is said to be cyclic if $c\in \cC$ implies the cyclic shift of $c$ is also in $\cC$. With any vector $v$ we associate the polynomial $v(x)=\sum_{i=0}^{n-1} v_i x^i$. It is well known~\cite{MacSlo78} that for a cyclic code there exists a unique generator polynomial, $g(x)\in\F_2[x]$, $\deg(g)=r$, such that $c\in\cC$ if and only if $c(x)=u(x)g(x)$, for some $u(x)\in\F_2[x]$, $\deg(u)\leq n-r-1$. 

A binary primitive Melas code of length $2^k-1$, which we denote by $\Melas(k)$, is defined in the following way: given $\alpha\in \F_{2^k}$ primitive, the parity-check matrix is 
\begin{equation}
\label{eq:pcmelas}
H_k = \begin{pmatrix}
    1 & \alpha & \alpha^2 & \dots & \alpha^{2^k-2}\\
    1 & \alpha^{-1} & \alpha^{-2} & \dots  & \alpha^{-(2^k-2)}
\end{pmatrix},
\end{equation}
where each entry is actually a binary column vector with $k$ entries, obtained by choosing a basis of $\F_{2^k}$ over $\F_2$. If $k\geq 3$, the generator polynomial of $\Melas(k)$ is $M_{\alpha}(x)M_{\alpha^{-1}}(x)$ and has degree $2k$, making $\Melas(k)$ a $[2^k-1,2^k-2k-1]$ code.

By \cite[Chapter 7]{MacSlo78}, the codewords of the dual code, $\Melas(k)^{\perp}$, are precisely the LFSR sequences with connection polynomial $M_\alpha M_{\alpha^{-1}}$. Thus, by Theorem \ref{thm:LFSRtracecharacterization}, the codewords in $\Melas(k)^{\perp}$ are of the form
\[
(\Tr(\beta \alpha^j+\gamma \alpha^{-j}))_{0\leq j\leq 2^k-2},
\]
for some $\beta,\gamma\in\F_{2^k}$. Additionally, by Remark \ref{rem:LFSRpolycorrespondence}, there is a one-to-one correspondence between codewords of $\Melas(k)^\perp$ (distinct up to cyclic shift) and the polynomial orbits defined by the relation $\sim$.

The binary primitive $2$-error-correcting BCH code of length $2^k-1$ is defined in the following way: given $\alpha\in \F_{2^k}$ primitive, the parity-check matrix is defined as
\[
 \begin{pmatrix}
    1 & \alpha & \alpha^2 & \dots & \alpha^{2^k-2}\\
    1 & \alpha^3 & \alpha^6 & \dots  & \alpha^{3(2^k-2)}\\
\end{pmatrix}.
\]
We denote these codes $\BCH(2,k)$. If $k\geq 3$, the generator polynomial is $M_{\alpha}(x)M_{\alpha^3}(x)$ and has degree $2k$, see \cite[Chapter 9]{MacSlo78}.

\subsection{The $b$-symbol weight}
Now we introduce the definition of $b$-symbol weight for binary vectors. Recall that the Hamming weight of a vector $v=(v_0,\dots,v_{n-1})\in\F_2^n$ is defined as $\wt(v)\eqdef \#\set{0\leq i\leq n-1: v_i \neq 0}$. For each $1\leq b\leq n-1$, define the map $\pi_b:\F_2^n\rightarrow (\F_2^b)^n$ as
\[
\pi_b(v_0,\dots,v_{n-1})=\parenv*{(v_0,\dots,v_{b-1}),(v_1,\dots,v_b),\dots, (v_{n-1},v_0,\dots,v_{b-2})}.
\]
The $b$-symbol Hamming weight of $v\in\F_2^n$ is defined as the Hamming weight of $\pi_b(v)$. In particular, for $b=1$ we retrieve the regular Hamming weight of $v$. If we interpret $(v_0,\dots,v_{n-1})$ as a sequence, the $b$-symbol weight is equal to $n$ minus the number of occurrences of the pattern $0^b$ (a run of $b$ zeros) in the sequence (where we consider the last symbol to be consecutive to the first). 

\subsection{Burst-covering radius and critical exponent}
A new parameter on codes, the burst-covering radius, was introduced in \cite{SacSch26}. In the particular case of linear codes, this parameter can be defined in the following way.

\begin{definition}
    Let $\cC$ be an $[n,n-r]$ binary linear code, and $H\in\F_2^{r\times n}$ a parity-check matrix for $\cC$. The burst-covering radius, $b(\cC)$,  is the least integer $b$ satisfying that every column vector $x\in\F_2^r$ can be obtained as a linear combination of at most $b$ (cyclically\footnote{The last and first columns of $H$ are also considered consecutive.}) consecutive columns of $H$.
\end{definition}

It was shown in~\cite{SacSch26} that $b(\cC)$ is independent of the choice of parity-check matrix, $H$. For cyclic codes, there exists a useful characterization of the burst-covering radius in terms of polynomial orbits.

\begin{theorem}[{\cite[Theorem 23]{SacSch26}}]\label{thm:cyclicburstcovering}
    Let $\cC$ be an $[n,n-r]$ binary cyclic linear code with generator polynomial $g\in\F_2[x]$. Then
    \[
    b(\cC)= \max_{\substack{f\in \F_2[x]\\ \deg (f)<r, f\neq 0}} \min_{j\geq 0}\deg (x^j f\pmod{g})+1.
    \]
\end{theorem}

As a consequence of the last theorem, in the case $\cC= \Melas(k)$ we can use our previously introduced notation and write the burst-covering radius as,
\[
     b(\cC)= \max_{\substack{f\in \F_2[x]\\ \deg f<r, f\neq 0}} \min_{f'\in \cO_f}\deg f'+1.
\]

Finally, the critical exponent of a linear code is defined in the following way:

\begin{definition}[\cite{Gre67}]
 The critical exponent of an $[n,n-r]$ binary linear code $\cC\subseteq \F_2^n$, denoted $c(\cC)$, is the smallest integer $c$ such that $\cC$ contains $c$ codewords whose supports cover all $n$ coordinate positions. If $\cC$ has a coordinate which is always zero, the critical exponent is defined to be $\infty$.
\end{definition}

In \cite{SacSch26}, the following connection between the critical exponent and the burst-covering radius of a cyclic code was established:

\begin{theorem}\label{thm:critexpbound}
    Given $\cC$ a binary $[n,n-r]$ cyclic code with burst-covering radius $b$, $\cC\neq \F_2^n$, the critical exponent of its dual code satisfies:
    \[
        c(\cC^{\perp})\leq r-b+1.
    \]
\end{theorem}

\section{Pattern frequencies and $k$-symbol weight}\label{sec:patternfreqandweight}
In this section we will show how the connection between LFSR sequences and polynomial orbits can be used to give improved bounds on the number of occurrences of patterns in $\Melas(k)^\perp$ codewords and their $k$-symbol weights. Throughout this section we fix some $\alpha\in\F_{2^k}$, $k\geq 3$, $\alpha$ primitive, so that the equivalence relation $\sim$ is defined.

\subsection{Lower bound on $k$-symbol weight}

Let us define the following two sets,
\begin{equation}
\label{eq:defHU}
\begin{split}
H &\eqdef \set*{ h\in\F_2[x] : \text{$h$ is irreducible and self-reciprocal}},\\
U &\eqdef \set*{ \set{u,\rev{u}}: u\in\F_2[x], \deg(u)>1, \text{$u$ is irreducible and not self-reciprocal}}.
\end{split}
\end{equation}
We will often abuse notation and use $u\in U$ to indicate choosing one arbitrary representative from a pair in $U$.

\begin{lemma}\label{lemma:orbitcount}
Let $\alpha\in\F_{2^k}$ be primitive, $k\geq 3$. Let $P\in \F_2[x]$ with $x\nmid P$ and $\deg (P)=d<k-s$ for some $s\geq 0$. Suppose that $P$ factors as product of irreducible polynomials in the following way,
    \[
    P = \prod_{h\in H} h^{a_h} \prod_{u\in U} u^{b_u} \rev{u}^{c_u}.
    \]
    Then, the number of polynomials $g\sim P$ with $\deg (g)<k-s$ is 
    \begin{equation}\label{eq:orbitcount}
        (k-s-d)\prod_{u\in U}(1+b_u+c_u).
    \end{equation}
\end{lemma}

\begin{IEEEproof}
Since $\deg(P)<k$, $\cO_P$ is non-degenerate. By Lemma \ref{lemma:equiv}, $g\sim P$ if and only if $g(\alpha)g(\alpha^{-1})=P(\alpha)P(\alpha^{-1})$. Let $g=x^bQ$, $x\nmid Q$, $\deg (g)<k-s$. By Lemma \ref{lemma:orbitequivdef}, Claim 1), $g$ belongs to the same orbit as $P$
if and only if $Q(x)\overleftarrow{Q}(x)=P(x)\overleftarrow{P}(x)$. If $Q$ factors as
    \[
    Q = \prod_{h\in H} h^{a'_h}\prod_{u\in U} u^{b'_u} \rev{u}^{c'_u},
    \]
then by uniqueness of factorization this is equivalent to asking $b'_u+c'_u=b_u+c_u$ for every $u\neq \rev{u}$, and $a_h=a'_h$ for every $h\in H$. There are $1+b_u+c_u$ choices for $(b'_u,c'_u)$ satisfying this condition. On the other hand, $\deg (Q)=\deg (P)$, and $\deg (g)<k-s$, thus $0\leq b< k-s-\deg (P)$. Putting all these together yields the desired result.
\end{IEEEproof}

\begin{theorem}\label{thm:occurofzero}
Consider any non-zero codeword $c\in\Melas(k)^\perp$, $k\geq 3$. Suppose $s_k\geq 0$ are such that $k-s_k\rightarrow \infty$. For the sake of simplicity of notation, we omit the dependency of $s$ on $k$. Then the number of occurrences of the pattern $0^{k+s}$ in $c$ is upper bounded by $2^{\frac{2(k-s)}{\log_2 (k-s)}(1+o(1))}$. Furthermore, there exists a non-zero codeword of $\Melas(k)^\perp$ with at least $2^{\frac{k-s}{2 \log_2 (k-s)}(1-o(1))}$ occurrences of $0^{k+s}$.
\end{theorem}

\begin{IEEEproof}
As discussed in Section~\ref{sec:prelim}, any codeword of $\Melas(k)^\perp$ is an LFSR sequence that is the output of a Galois-mode LFSR with connection polynomial $M_\alpha M_{\alpha^{-1}}$ and an initial load $f\in\F_2[x]$, $\deg(f)<2k$, where $\alpha\in\F_{2^k}$ is primitive. As the claim omits the all-zero codeword, $f\neq 0$. By Theorem \ref{thm:characterizationzerosdeg}, a run of $k+s$ zeros in this LFSR sequence corresponds to a polynomial of degree $\leq k-s-1$ in $\cO_f$. Thus, we need to count the number of these polynomials. 

Assume $P\in\cO_f$, $\deg(P)\leq k-s-1$, and $x\nmid P$. Since $f\neq 0$, and since $\deg(P)<k$, $\cO_f$ is non-degenerate. Suppose $P$ factors as in Lemma \ref{lemma:orbitcount}. We will now bound the value of \eqref{eq:orbitcount}. The first factor satisfies
\[
k-s-\deg(P) \leq k-s = 2^{\log_2 (k-s)}.
\]
Let $D \eqdef \floor{\frac{1}{2}\log_2 (k-s)}$. The number of irreducible factors $u$ of degree less than $D$ is at most $2^{D}\leq \sqrt{k-s}$. These contribute at most 
\[
\prod_{u: \deg(u)<D} (1+b_u+c_u) \leq \prod_{u: \deg(u)<D} (k-s) \leq (k-s)^{\sqrt{k-s}}= 2^{\sqrt{k-s}\log_2(k-s)}.
\]
For those factors $u$ with degree greater or equal to $D$, we have that
$D\cdot\sum_{u: \deg (u)\geq D} (b_u+c_u)\leq k-s$, and so 
\[\prod_{u: \deg (u)\geq D} (1+b_u+c_u)\leq \prod_{u: \deg (u)\geq D} 2^{b_u+c_u}\leq 2^{(k-s)/D}\leq 2^{(k-s)/(1/2 \log_2(k-s)-1)}=2^{2(k-s)/(\log_2 (k-s)-2)}.\]
By multiplying together all three bounds we obtain the claimed result.

    For the existence result, choose $d$ odd such that $2\log_2 (k-s) \leq d\leq 2\log_2 (k-s)+2$. The number of binary irreducible polynomials of degree $d$ is given by (see \cite{LidNie97})
    \[
    \frac{1}{d}\sum_{e|d}\mu(e)2^{d/e}\geq\frac{2^d-d 2^{d/3}}{d},
    \]
    where $\mu$ is the M\"obius function, and the bound follows because every proper divisor of $d$ odd is at most $d/3$. For large enough $k-s$, $2^d-d2^{d/3}>2(k-s-1)+d$, and thus there are at least $\frac{2(k-1-s)}{d}+1$ distinct irreducible polynomials of degree $d$. It is known that no odd-degree irreducible polynomial is self-reciprocal except for $x+1$ (e.g., see~\cite[Remark 2.1.50]{MulPan13}), so we deduce that there exist at least $r= \lfloor\frac{k-s-1}{d}\rfloor$ pairs of reciprocal irreducible polynomials of degree $d$, $\set{u_1,\rev{u_1}},\dots, \set{u_r, \rev{u_r}}$. 

    Consider the $2^r\geq 2^{\frac{k-s-1}{2\log_2 (k-s)+2}-1}$ polynomials
    \[
    \prod_{i=1}^r u_i^{b_i}\rev{u_i}^{1-b_i},
    \]
    with $b_1,\dots,b_r\in\{0,1\}$. They all have degree $r\cdot d<k-s$ and they all belong to the same orbit by the proof of Lemma \ref{lemma:orbitcount}. 
\end{IEEEproof}

\begin{corollary}\label{cor:lowboundkweight}
The $k$-symbol weight of any non-zero codeword in $\Melas(k)^{\perp}$, $k\geq 3$, is at least $2^k-1-2^{\frac{2k(1+o(1))}{\log_2 k}}$.
\end{corollary}

The bounds in Theorem \ref{thm:occurofzero} and Corollary \ref{cor:lowboundkweight} are strictly better than the best known results in the literature. Theorem 30 in \cite{SacSch26} (or alternatively, an adaptation of Theorem V.4 in \cite{ShiOzb21}) implies an upper bound on the number of occurrences of $0^k$ of the order $2^{k/2}$, in contrast to $2^{k/\log k}$. The computation in the proof of \cite[Theorem 27]{SacSch26} shows that the number of occurrences $N_{0^k}$ of $0^k$ in the LFSR sequence $(\Tr(\alpha^j+c\cdot \alpha^{-j}))_{0\leq j\leq 2^k-2}$, $c\in\F_{2^k}^*$, is given by
\[
N_{0^k} = \frac{1}{2^k}\sum_{J\subseteq [k-1]} \varphi(J),
\]
where 
\[
\varphi(J) = \sum_{x\in\F_{2^k}^*}\chi\parenv*{J(\alpha)\cdot x+c\cdot J(\alpha^{-1})\cdot x^{-1}},
\]
 $\chi$ is the additive character $\chi(x)=(-1)^{\Tr(x)}$, $[k-1]=\set{0,1,\dots,k-1}$, and $J\subseteq [k-1]$ is interpreted as being the non-zero coefficients of a polynomial of degree $<k$. 

 More generally, for any $w\in\F_2^k$, the number of occurrences $N_w$ is given by
 \begin{equation}\label{eq:fourier}
 N_w=\frac{1}{2^k}\sum_{J\subseteq [k-1]}(-1)^{\sum_{j\in J}w_j}\varphi(J).     
 \end{equation}


 It is known by the Weil-Carlitz-Uchiyama bound (see for example \cite{MorMor91}), that $|\varphi(J)|\leq 2+2^{1+k/2}$ if $J\neq \emptyset$. This yields the upper  bound $|N_{0^k}|\leq 1+ 2^{1+k/2}$, which is strictly worse than Theorem \ref{thm:occurofzero}. That means that there is significant cancellation in the computation of $N_{0^k}$, which a direct application of the Weil bound misses. 

So far we have only focused on the all-zero pattern. Can we extend the previous results to others? The following theorem shows that the same techniques allow us to give an answer to this question to a somewhat unsatisfactory extent.

\begin{theorem}\label{thm:nonzeropatterns}
Consider any codeword $c\in\Melas(k)^\perp$, $k\geq 3$. There exist at least $2^{\floor{\frac{k}{\log_2 k}}+1}$ patterns $w\in\F_2^k$ such that they occur at most $2^{\frac{5k}{\log_2 k}(1+o(1))}$ times in $c$.
\end{theorem}

\begin{IEEEproof}
Once again, any codeword of $\Melas(k)^\perp$ is an LFSR sequence that is the output of a Galois-mode LFSR with connection polynomial $M_\alpha M_{\alpha^{-1}}$ and an initial load $f\in\F_2[x]$, $\deg(f)<2k$, where $\alpha\in\F_{2^k}$ is primitive. The claim trivially holds for the all-zero sequence, so for the rest of the proof we assume $f\neq 0$. 

Let $s_k=\floor{\frac{k}{\log_2 k}}$. Consider $w\in \F_2^k$ such that its corresponding polynomial, $h_w$, as defined in Theorem \ref{thm:characterizationotherpatterns}, has degree $\deg (h_w)\leq k+s_k$. There are $2^{s_k+1}$ such $w$, since the $h_w$ are multiples of $x^k$ and have degree $\leq k+s_k$. By Theorem \ref{thm:characterizationotherpatterns}, the number of occurrences of $w$ is equal to the number of polynomials in the orbit which also belong to the affine space $\langle 1,x,\dots,x^{k-1}\rangle +h_w$, i.e.,
\[\#(\cO_f \cap \langle 1,x,\dots,x^{k-1}\rangle +h_w) \leq \# \set*{ u\in \cO_f : \deg(u) \leq k+s_k}.
\]

Assume $P\in\cO_f$, $\deg(P)\leq k+s_k$, $x\nmid P$, and $P$ is of maximal degree among all valid choices. By Lemma \ref{lemma:orbitequivdef}, Claim 2), any other polynomial in $\cO_f$, $g=x^b Q$ with $x\nmid Q$ and $\deg g\leq k+s_k$, has to satisfy one of the following equations:
    \[
    P(x)\rev{P}(x)+M_\alpha(x)M_{\alpha^{-1}}(x) R(x)=x^{\deg (P)-\deg (Q)} Q(x)\rev{Q}(x) ,
    \]
    with $\deg(R)\leq 2s_k$ and $x^{2(\deg P-k)}R(x^{-1})=R(x)$. Since $P$ is fixed, for each of the possible $2^{s_k+1}$ polynomials $R$ we can compute the left-hand side of the above equation and compute its factorization,
    \[
   x^{a_R}\prod_{h\in H}h^{a_{h,R}}\prod_{u\in U}u^{b_{u,R}}\rev{u}^{c_{u,R}}=   x^{\deg P-\deg Q} Q(x)\rev{Q}(x),
    \]
    where the left-hand side has degree at most $2(k+s_k)$. These equations might or might not have a solution: all the exponents $a_{h,R}$ must be even and $b_{u,R}=c_{u,R}$ has to hold for every $u\in U$ for there to exist a valid $Q$. If there are solutions, by the same argument as in the proof of Lemma \ref{lemma:orbitcount}, the number of possible polynomials $Q$ is $\prod_{u\in U}(1+b_{u,R})$. For each of the $2^{s_k+1}$ equations, the exponent $b$ can take at most $k+s_k+1$ values, so the total number of solutions is upper bounded by
    \[
    \sum_{\substack{R : \deg (R)\leq 2s_k \\x^{2(\deg P-k)}R(x^{-1})=R(x)}} (k+s_k+1)\prod_{u\in U}(1+b_{u,R}).
    \]
    Notice that $\sum_{u\in U}\deg u \cdot b_{u,R}\leq k+s_k$. Following the same proof as in Theorem \ref{thm:occurofzero}, we can give the bound
    \[
    2^{s_k+1}\cdot (k+s_k+1)\cdot 2^{\sqrt{k+s_k}\log_2(k+s_k)}\cdot 2^{2(k+s_k)/(\log_2(k+s_k)-2)}\leq 2^{\frac{5k}{\log_2 k}(1+o(1))}.
    \]
\end{IEEEproof}

We conjecture that a similar bound on the number of occurrences should hold for every pattern $w\in \F_2^k$, and evidence supporting this claim is discussed in Section \ref{sec:connection}.

\begin{conjecture}\label{conj:klogk}
The number of occurrences of any $w\in \F_2^k$ in an arbitrary non-zero codeword of $\Melas(k)^\perp$ is at most $2^{\frac{\ell k}{\log_2 k}(1+o(1))}$ for some constant $\ell$.
\end{conjecture}

It would seem that the techniques we used so far are  incapable of proving this conjecture. While it might be tempting to attribute this to some deficiency in the proposed approach, in Section \ref{sec:connection} we actually show that there really exists a fundamental difference between the behavior of the all-zero pattern $0^k$ and other patterns.

To conclude this section, we show an analogous result to Theorem \ref{thm:occurofzero} for runs of zeros of length $<k$.

\begin{theorem}\label{thm:occurshorterzeros}
Consider any non-zero codeword $c\in\Melas(k)^\perp$. Let $s_k\geq 0$ for each $k$. For the sake of simplicity of notation, we omit the dependency of $s$ on $k$. The number of occurrences of the pattern $0^{k-s}$ in $c$ is upper bounded by $2^{s+\frac{2(k+s)(1+o(1))}{\log_2 (k+s)}}$. Furthermore, there exists a non-zero codeword of $\Melas(k)^\perp$ with at least $2^{\frac{k+s}{2 \log_2 (k+s)}(1-o(1))}$ occurrences of $0^{k-s}$.
\end{theorem}

\begin{IEEEproof}
$c$ is an LFSR sequence that is the output of a Galois-mode LFSR with connection polynomial $M_\alpha M_{\alpha^{-1}}$ and an initial load $f\in\F_2[x]$, $\deg(f)<2k$, $f\neq 0$, where $\alpha\in\F_{2^k}$ is primitive. By Theorem \ref{thm:characterizationzerosdeg}, the number of occurrences of $0^{k-s}$ in $c$ is equal to $\#\set{u\in\cO_f : \deg(u)\leq k+s-1}$. The proof is now the same as the proof of Theorem \ref{thm:nonzeropatterns}, but the $2^{s+1}$ is not absorbed in the final bound.

For the existence result, the proof is the same as in Theorem \ref{thm:occurofzero}. (This shows that the equation corresponding to $R(x)=0$ has enough solutions, while disregarding the rest of the equations).
\end{IEEEproof}

As a final corollary for this subsection, we show that Theorem~\ref{thm:occurshorterzeros} implies Conjecture~\ref{conj:klogk} for patterns of low linear complexity.

\begin{corollary}
Let $w=(w_0,\dots,w_{k-1})\in\F_2^k$ be a binary pattern, and assume there exist $0\leq s<k$ and coefficients not all zero $a_0,\dots,a_s\in\F_2$ such that for all $0\leq j< k-s$, $\sum_{i=0}^s a_i w_{j+i} = 0$. Then the number of occurrences of $w$ in any non-zero codeword $c\in\Melas(k)^\perp$ is upper bounded by $2^{s+\frac{2(k+s)(1+o(1))}{\log_2 (k+s)}}$. In particular, for patterns with $s=\cO(\frac{k}{\log_2 k})$, Conjecture~\ref{conj:klogk} holds.
\end{corollary}

\begin{IEEEproof}
Let $c\in\Melas(k)^\perp$ be a non-zero codeword, $c=(c_0,\dots,c_{2^k-2})$. Let $E$ be the cyclic rotation operator, i.e., $E(c)=(c_1,\dots,c_{2^k-2},c_0)$. Note that any occurrence of $w$ in $c$ induces an occurrence of $0^{k-s}$ in $\sum_{i=0}^s a_i E^i(c) = c'$, where $c'\in\Melas(k)^\perp$. Since $s<k$, necessarily $c'\neq 0$. The conclusion follows by Theorem~\ref{thm:occurshorterzeros}.
\end{IEEEproof}

\subsection{Codewords with full $k$-symbol weight}

Now we shift the focus from studying individual codewords in $\Melas(k)^\perp$ and the number of occurrences of $0^k$, to analyzing the number of codewords that have this pattern at least once. This will allow us to give asymptotics on the number of codewords with full $k$-symbol weight, and also give an improved lower bound on the burst-covering radius of Melas codes in the next section.

Recall the definition of $H$ and $U$ from~\eqref{eq:defHU}. Consider the following set of tuples of natural numbers:
\[
S = \set{((a_h)_{h\in H}, (n_u)_{u\in U}): \#\supp(a_h),\#\supp(n_u)<\infty}.
\]
We define the \emph{degree} of an element in $S$ as $\sum_{h\in H}a_h\deg(h)+\sum_{u\in U}n_u\deg(u)$. With each non-zero $f\in\F_2[x]$, we associate an element $\phi(f)\in S$ in the following way: 
\[
f= x^a \prod_{h\in H}h^{a_h}\prod_{u\in U}u^{b_u}\rev{u}^{c_u} \mapsto \phi(f)\eqdef ((a_h)_{h\in H}, (b_u+c_u)_{u\in U}).
\]
Notice that $\deg(\phi(f))+a=\deg(f)$, so that $\deg(\phi(f))\leq \deg(f)$.

\begin{lemma}\label{lemma:phiequality}
Let $\alpha\in\F_{2^k}$ be primitive, $k\geq 3$. Let $f,g\in\F_2[x]$, $\deg(f),\deg(g)<2k$, be such that $f(\alpha),f(\alpha^{-1})\neq 0$. If $\phi(f)=\phi(g)$, then $f\sim g$. Whenever $\deg(f),\deg(g)<k$, the converse is also true.
\end{lemma}

\begin{IEEEproof}
Let $\phi(f)=((a_h)_{h\in H}, (n_u)_{\set{u,\rev{u}}\in U})$. Notice that $f(x)f(x^{-1})=x^{-\deg(\phi(f))}\prod_{h\in H}h^{2a_h}\prod_{u\in U}(u\rev{u})^{n_u}$ is completely determined by $\phi(f)$, so $f(x)f(x^{-1})=g(x)g(x^{-1})$, which clearly implies that $f(\alpha)f(\alpha^{-1})=g(\alpha)g(\alpha^{-1})$, and thus $f\sim g$, by Lemma~\ref{lemma:equiv}.

If $\deg(f),\deg(g)<k$, and $f\sim g$ belong to a non-degenerate orbit, then by Lemma \ref{lemma:orbitequivdef}, $f(x)f(x^{-1})=g(x)g(x^{-1})$, and thus $\phi(f)=\phi(g)$.
\end{IEEEproof}

Define
\begin{align*}
\sigma(d)&\eqdef \#\set*{s\in S: \deg(s)=d},\\
\tau(d) &\eqdef \#\set*{ h\in H : \deg(h)=d}.
\end{align*}
From Lemma \ref{lemma:phiequality} we immediately obtain the following theorem.
\begin{theorem}\label{thm:orbitupbound}
Let $k\geq 3$, and define $\theta(D)$ to be the number of distinct non-degenerate orbits of $\sim$ containing a polynomial of degree $<D$. Then, 
\[
\theta(D)\leq \sum_{d< D}\sigma(d).
\]
    Furthermore, if $D\leq k$, then equality holds.
\end{theorem}

\begin{IEEEproof}
    Given an orbit which contains an element $f$ of degree $< D$, we can associate to it the element $\phi(f)\in S$ which has degree $\deg(\phi(f))\leq \deg(f)< D$. Notice that there might be more than one choice for each such orbit, but given two distinct non-degenerate orbits, their associated elements in $S$ will be different thanks to Lemma \ref{lemma:phiequality}. This proves the desired inequality.

    If $D\leq k$, define the mapping $\cO_f\mapsto \phi(f)$. By Lemma \ref{lemma:phiequality}, it is well defined and injective. Given $s=((a_h)_{h\in H}, (n_u)_{u\in U})\in S$ of degree $d<D$,  the polynomial $f=\prod_{h\in H}h^{a_h}\prod_{u\in U}u^{n_u}$ has degree $d<D$, $\phi(f)=s$, and the orbit of $f$ is non-degenerate since $\deg(f)<k$. This shows that the map is also surjective.   
\end{IEEEproof}

Define the following three generating functions
\[
\cH(z)= \prod_{h\in H}\frac{1}{1-z^{\deg h}}, \qquad \cU(z)=\prod_{u\in U}\frac{1}{1-z^{\deg u}}, \qquad \Sigma(z)=\sum_{d\geq 0} \sigma(d) z^d.
\]

\begin{lemma}\label{lemma:genfuncidentities}
    The previous generating functions satisfy,
    \begin{itemize}
        \item [a)]$\cH(z)\cU(z)=\Sigma(z)$,
        \item [b)]$\cH(z)\cU^2(z)=(1-z)/(1-2z)$,
        \item [c)] $\cH(z)$ is analytic in $|z|<1/\sqrt{2}$ and has a singularity in $z=1/\sqrt{2}$.
    \end{itemize}
       
\end{lemma}

\begin{IEEEproof}
    See the Appendix. 
\end{IEEEproof}

For $\lambda\in \C$ we can define the formal exponentiation by:
\[
    (1+F(z))^{\lambda}=\sum_{n\geq 0}\binom{\lambda}{n}F(z)^n
    \]
    (see for example \cite[Example 1.1.10]{Sta11}), which allows us to define the formal square root. Define
\[\cA(z)\eqdef \sqrt{(1-z)\cH(z)}=\sum_{n\geq 0}a_n z^n.\]
    
\begin{theorem}\label{thm:genfuncbigequality}
We have
\[
    \Sigma(z)= \sqrt{\frac{(1-z)\cH(z)}{1-2z}}= (1-2z)^{-1/2}\cA(z),
\]    
where we choose the branch of the square root such that $\Sigma(0)=1$.
\end{theorem}
\begin{IEEEproof}
The proof is immediate from Lemma \ref{lemma:genfuncidentities}.
\end{IEEEproof}

\begin{theorem}\label{thm:sigmaasymptotics}
\[
\sigma(d)=\cA(1/2)\frac{2^d}{\sqrt{\pi d}}(1+\cO(1/d))
\]    
\end{theorem}

\begin{IEEEproof}
   See the Appendix.
\end{IEEEproof}

\begin{theorem}\label{thm:sigmasumestimate}
    \[
    \sum_{d<D} \sigma(d)=\cA(1/2)\frac{2^D}{\sqrt{\pi D}}(1+\cO(1/D)).
    \]
\end{theorem}

\begin{IEEEproof}
   See the Appendix.
\end{IEEEproof}

\begin{theorem}\label{thm:notfullsymbolweight}
    Let $s\geq 0$ such that $k-s\rightarrow \infty$. The number of codewords in $\Melas(k)^{\perp}$ which do not have full $(k+s)$-symbol weight is
    \[
   \cA(1/2)\frac{2^{2k-s}}{\sqrt{\pi (k-s)}}\parenv*{1+\cO\parenv*{\frac{1}{k-s}}}.
    \]    
    In particular, the number $A^{(k+s)}_{2^k-1}$ of full $(k+s)$-symbol weight codewords satisfies 
    \[
    \frac{A^{(k+s)}_{2^k-1}}{2^{2k}}=1-\cO\parenv*{\frac{1}{2^s\sqrt{k-s}}}\longrightarrow 1.
    \]
\end{theorem}

\begin{IEEEproof}
    By Theorem \ref{thm:characterizationzerosdeg}, a codeword in $\Melas(k)^{\perp}$ contains a run of $k+s$ zeros (and thus, does not have full $(k+s)$-symbol weight) if and only if the corresponding polynomial orbit contains a polynomial of degree $< 2k-(k+s)=k-s$. Combine the equality from Theorem \ref{thm:orbitupbound} and Theorem \ref{thm:sigmasumestimate}. Since each non-degenerate orbit induces $2^k-1$ codewords (there are $2^k-1$ distinct rotations), we obtain the claimed result.
\end{IEEEproof}

\section{Burst-covering radius and critical exponent}\label{sec:burstcoveringradius}

In this section, we prove new bounds on the burst-covering radius of $\Melas(k)$ as well as the critical exponent of $\Melas(k)^\perp$. These improve upon known results by using tools from the previous section. We then also provide an improved algorithm for computing the exact burst-covering radius of Melas codes. We start with the burst-covering radius.

\begin{theorem}
The burst-covering radius of $\Melas(k)$ satisfies
    \begin{align*}
    b(\Melas_k) &
    \geq k+\frac{1}{2}\log_2 k + \log_2\parenv*{\frac{\sqrt{\pi}}{\cA(1/2)}} -o(1) \\   
    &\geq k+\frac{1}{2}\log_2 k +0.550958\ldots -o(1).
    \end{align*}
\end{theorem}

\begin{IEEEproof}
By Lemma \ref{lem:orbitsize}, the number of non-degenerate orbits is $2^k-1$. By Theorem \ref{thm:orbitupbound}, if $D$ is such that $\sum_{d< D}\sigma(d)\leq 2^k-2$, then there exists an orbit whose elements all have degree $\geq D$, and thus $b(\Melas_k)\geq D+1$ by Theorem \ref{thm:cyclicburstcovering}.

    By Theorem \ref{thm:sigmasumestimate}, it suffices to ask 
    \[
     D\leq k+\frac{1}{2}\log_2 D+\log_2\parenv*{\frac{\sqrt{\pi}}{\cA(1/2)}}-o(1),
    \]
    to guarantee $\sum_{d<D}\sigma(d)\leq 2^k-2$. For $D>k$ we can choose $D$ as big as
    \[
    D\leq k+\frac{1}{2}\log_2 k+\log_2\parenv*{\frac{\sqrt{\pi}}{\cA(1/2)}}-1-o(1) \leq k+\floor*{\frac{1}{2}\log_2 k+\log_2\parenv*{\frac{\sqrt{\pi}}{\cA(1/2)}}-o(1)}.
    \]
The value of $\log_2(\cA(1/2))$ can be estimated by noticing
    \[
    \log_2(\cA(1/2))=\frac{\log(\cA(1/2))}{\log 2}=\frac{0.5 \log 0.5+0.5\log \cH(0.5)}{\log 2}\approx 0.274789\dots,
    \]
    where $\log$ is the natural logarithm and  $\log \cH(0.5)= \sum_{d\geq 1}\tau(d)\log\parenv{\frac{1}{1-0.5^d}}$ was estimated by using the formula
    \[\tau(d)=\frac{1}{d}\sum_{j|\frac{d}{2}, j\text{ odd}}\mu(j)2^{\frac{d}{2j}},\]
    for $d$ even (see \cite{Car67}), $\tau(d)=0$ for $d$ odd greater than $1$, $\tau(1)=1$. 
\end{IEEEproof}

The previous bound~\cite[Theorem 38]{SacSch26} gave only $b(\Melas(k))\geq k+2$. Thanks to Theorem \ref{thm:critexpbound}, we also obtain the following bound on the critical exponent of the dual Melas codes.

\begin{corollary}
    The critical exponent of $\Melas(k)^{\perp}$ is upper bounded by
    \[
    c(\Melas(k)^\perp)\leq k-\frac{1}{2}\log_2k+0.449041\ldots+o(1).
    \]
\end{corollary}

As far as we are aware, this bound is strictly better than what could be obtained by known results, such as Bose-Burton's theorem \cite{BoBu66}.

\subsection{An improved algorithm to compute the burst-covering radius}
In \cite{SacSch26} a general algorithm is given for computing the burst-covering radius of binary cyclic codes. In this section we propose an improved algorithm for the case of Melas codes. This algorithm also applies to primitive double-error-correcting BCH codes, and can be further generalized to any binary cyclic code with simple roots, although we do not give this generalization as it lies outside the scope of this paper. 

Recall the parity-check matrix $H_k$ of~\eqref{eq:pcmelas} for $\Melas(k)$. Consider the syndromes of $\Melas(k)$, $(x_1,x_2)\in \F_{2^k}^2$. Define an equivalence relation among syndromes as follows: $(x_1,x_2)\approx(y_1,y_2)$ if and only if $y_1=\alpha^j x_1$, $y_2=\alpha^{-j} x_2$ for some integer $j$. Besides $\set{(0,0)}$, there are two degenerate equivalence classes $\set{(0,x_2): x_2\in \F_{2^k}}$, $\set{(x_1,0): x_1\in \F_{2^k}}$. The rest of the equivalence classes have $2^k-1$ elements each. These can be identified by a single element in $\F_{2^k}$: the product $x_1\cdot x_2$ is constant throughout each equivalence class, and two distinct non-degenerate equivalence classes have distinct values of $x_1\cdot x_2$. 

Consider a linear combination within a window of width $t$ starting at the first column of $H_k$. This linear combination can be represented with a polynomial $f\in\F_2[x]$, $\deg (f)<t$, and the syndrome is $(f(\alpha),f(\alpha^{-1}))$. As we slide the window, we generate all the syndromes $(\alpha^j f(\alpha),\alpha^{-j}f(\alpha^{-1}))$, that is, the equivalence class of the original syndrome. 

Algorithm \ref{alg:covering} computes the burst-covering radius of $\Melas(k)$ using the following strategy. The polynomials in $\F_2[x]$ are listed using a binary reflected Gray code. In this way, all polynomials of degree $j$ are listed before all polynomials of degree $j+1$, for every $j$. For each polynomial $f$ we compute the corresponding syndrome $(f(\alpha),f(\alpha^{-1}))$ and record that the corresponding syndrome equivalence class has been covered. To this end, we employ two boolean variables corresponding to the degenerate orbits $\texttt{degenerate\_1, degenerate\_2}$, and a boolean array $\texttt{non\_degenerate\_orbs}$ indexed by the elements of $\F_{2^k}$. In terms of implementation, the latter can be thought of as a boolean array of dimensions $2\times 2\times\dots \times 2$, indexed by the coefficients of the polynomial representation of the field elements. The algorithm finishes when all orbits have been covered, and outputs the current size of the window (which can be computed from the current index in the Gray code ordering of the polynomials). 

The algorithm does not actually generate the full Gray code listing at the beginning, but instead employs the function $\texttt{GrayTransition(i)}$, which given an index in the Gray code, computes the bit position which needs to be flipped in that step. In that way, the syndromes can be updated easily by adding $(\alpha^j,\alpha^{-j})$, where $j$ is the bitflip position. The $i$-th vector in the reflected binary Gray code can be computed with bitwise operations as $i\oplus \floor{i/2}$ (the bitwise XOR of the binary representations), see \cite[Eq. 7.2.1.1–(9)]{Knu11}, and $\texttt{GrayTransition(i)}$ can be computed from this. We can assume it requires at most $\cO(k)$ time.

\begin{algorithm}[t]
\caption{Computation of the burst-covering radius of $\Melas(k)$}\label{alg:covering}
\begin{algorithmic}[1]
\State \textbf{Input:} $k\geq 3, \alpha\in\F_{2^k}$ \text{ primitive}
\State \textbf{Output:} $b$ \text{ burst-covering radius of} $\Melas(k)$ \text{with generator polynomial} $M_\alpha M_{\alpha^{-1}}$
\Function{Burst\_radius\_Melas}{$k, \alpha$}
    \State $\texttt{non\_degenerate\_orbs}\gets \underbrace{[\texttt{False},\dots, \texttt{False}]}_{2^k-1}$ 
    \State $\texttt{degenerate\_1}\gets \texttt{False}$
    \State $\texttt{degenerate\_2}\gets \texttt{False}$
    \State $\texttt{covered\_orbits}\gets 0$
    \State $\texttt{synd}\gets (1,1)$
    \State $\texttt{i}\gets 1$
    \While{\texttt{True}}
        \If{$\texttt{synd[0]=0 and not degenerate\_1}$}
           \State $\texttt{degenerate}\_1\gets \texttt{True}$
           \State $\texttt{covered\_orbits++}$
        \ElsIf{$\texttt{synd[1]=0 and not degenerate\_2}$}
           \State $\texttt{degenerate}\_2\gets \texttt{True}$
           \State $\texttt{covered\_orbits++}$
        \ElsIf{$\texttt{synd[0]}\neq 0 \texttt{ and synd[1]}\neq 0$}
            \State $\texttt{orbit}\gets \texttt{synd[0]}\cdot \texttt{synd[1]}$
            \If{\texttt{not non\_degenerate\_orbs[orbit]}}
                \State $\texttt{non\_degenerate\_orbs[orbit]}\gets \texttt{True}$
                \State $\texttt{covered\_orbits++}$
            \EndIf    
        \EndIf
        \If{$\texttt{covered\_orbits}=2^k+1$}
            \State \Return $\floor{\log_2 i}+1$
        \EndIf
        \State $j\gets \texttt{GrayTransition(i)}$
        \State $\texttt{synd}\gets (\texttt{synd}[0]+\alpha^j, \texttt{synd}[1]+\alpha^{-j})$
        
        \State $\texttt{i}++$
    \EndWhile
\EndFunction
\end{algorithmic}
\end{algorithm}

Now we analyze the time complexity of Algorithm \ref{alg:covering}. The initialization of variables requires $\cO(2^k)=\cO(n)$ time. The number of iterations is at most $2^{b(\Melas(k))}\leq 2^{\frac{3}{2}k+\frac{3}{2}}=\cO(n^{1.5})$ by \cite[Theorem 36]{SacSch26}. In each iteration, all operations require constant time except for the finite field sum, multiplication and exponentiation operations. Addition requires time $\cO(k)=\cO(\log n)$. Even with a naive implementation, the cost of the finite field multiplication is $\cO(k^2)=\cO(\log^2 n)$. The powers of $\alpha$ and $\alpha^{-1}$ can be computed the first time they are needed and then stored. Furthermore, $\alpha^j$ is not needed until after all the powers $\alpha,\alpha^2,\dots,\alpha^{j-1}$ have been computed, so it can be obtained with just one multiplication $\alpha\cdot \alpha^{j-1}$. Notice that the highest power needed $\alpha^j$ has $j\leq b(\Melas(k))\leq \frac{3}{2}k+\frac{3}{2}$. Thus, the total cost is $\cO(n+n^{1.5}\cdot \log^2 n+\log^2 n\cdot \frac{3}{2}\log n)=\cO(n^{1.5}\cdot \log^2 n)$. In comparison, Algorithm 1 in \cite{SacSch26} requires $\cO(n^2 \log n)$ time.

Algorithm \ref{alg:covering} was used to compute all possible values of the burst-covering radius of $\Melas(k)$ for $6\leq k\leq 12$, which can be found in Table \ref{tab:exhaustivemelas}. For each value of $b$, the table indicates how many non-conjugate primitive roots yield a code with that value of $b$, and one possible $M_\alpha$. 

\begin{table}
\caption{All possible burst-covering radii of $\Melas(k)$}
\label{tab:exhaustivemelas}
\begin{center}
\begin{tabular}{c|ccccc}
\hline\hline
    $k$ & Upper bound & $b$ &  $M_\alpha(x)$ & \# of roots  \\
    \hline 
    6 & 10 & 9 & $x^6+x+1$ & 2/6  \\
      &    & 10 & $x^6 + x^4 + x^3 + x + 1$ & 4/6 \\
      \hline
   7 & 12  & 11 & $x^7+x+1$ & 18/18   \\
   \hline
   8 & 13 & 12 & $x^8 + x^4 + x^3 + x^2 + 1$ & 16/16  \\
   \hline 
   9 & 15 & 13 & $x^9 + x^6 + x^4 + x^3 + 1$ & 26/48  \\
     &    & 14 & $x^9+x^4+1$ & 22/48\\
     \hline
   10 & 16 & 15 & $x^{10} + x^6 + x^5 + x^3 + x^2 + x + 1$ & 60/60  \\
   \hline
   11 & 18 & 16 & $x^{11}+x^2+1$ & 170/176 \\
      &    & 17 & $x^{11} + x^{10} + x^9 + x^8 + x^3 + x + 1$ & 6/176 \\    \hline
    12 & 19 & 17 & $x^{12} + x^7 + x^6 + x^5 + x^3 + x + 1$ & 132/144\\
       &    & 18 & $x^{12} + x^{10} + x^8 + x^4 + x^3 + x^2 + 1$ & 12/144\\
    \hline\hline
\end{tabular}    
\end{center}
\end{table}

\section{A Poisson randomness postulate}\label{sec:poisson}
\subsection{Introduction and definition}

In the study of infinite sequences over finite alphabets, several different notions of what it means for a sequence to behave randomly have been given. Probably the first and most famous one is Borel's \emph{normality} \cite{Borel1909}, which requires the frequencies of symbols to converge to the uniform distribution as we look at increasingly longer prefixes, as well as the frequencies of pairs of symbols, triplets, etc. If we only ask the frequencies of symbols to satisfy this requirement, the sequence is called \emph{simply normal}.

More recently, Zeev Rudnick proposed the definition of $\lambda$-\emph{Poisson generic sequences} based on his work on lacunary sequences \cite{RudZah99}. In particular, we will focus on the case $\lambda=1$. Let $S\in \F_2^{\N}$ be an infinite binary sequence. Denote by $N_w^{(k)}$ the number of occurrences of $w\in\F_2^k$ in the prefix of $S$ of length $2^k$. Consider the probability spaces $\Omega_k = \F_2^k$ with the uniform measure $\mu_k$, and define the random variables $X_k: \Omega_k \rightarrow \R$ as $X_k(w)=N_w^{(k)}$. Then $S$ is said to be \emph{$1$-Poisson generic} if
\[
X_k\xrightarrow[k\rightarrow \infty]{d} \cP(1),
\]
or equivalently, if
\[
\frac{\#\set{w\in \F_2^k: N_w^{(k)}=i}}{2^k}\xrightarrow[k\rightarrow \infty]{} \frac{e^{-1}}{i!}
\]
for every $i\geq 0$.

Peres and Weiss proved that almost all sequences with respect to the product measure are Poisson generic, although they did not publish their result. The proof was later presented in \cite{AlvBecMer23}, where it is also shown that there exist countably many computable Poisson generic sequences. An explicit construction of $\lambda$-Poisson generic sequences was given in \cite{BecSac23} for alphabets of size $3$ or larger, and in the binary case, the construction works for $\lambda \leq \ln(2)$.

In the context of finite (or periodic) binary sequences there has also been extensive work on the matter of randomness. For example, Golomb proposed three randomness postulates R1, R2 and R3 in his book on linear feedback shift register sequences \cite{Gol67}. R1 requires the number of ones and zeros in a period of the sequence to differ by at most $1$. Notice that this resembles the notion of simple normality that we mentioned earlier, albeit somewhat stronger. 

We propose a randomness postulate on a sequence of binary sequences, which extends the definition of $1$-Poisson generic sequences. Loosely speaking, we are looking for a sequence of binary sequences that asymptotically approaches the $1$-Poisson generic postulate.

\begin{definition}[\textbf{Poisson randomness postulate}]\label{def:poisson}
Let $\cS=(S^{(k)})_{k\in \N}$ be a sequence of infinite binary sequences. Let $N_w(S^{(k)})$ be the number of occurrences of $w\in\F_2^k$ whose first coordinate falls inside the prefix\footnote{Notice that by our convention, the windows of length $k$ that start in one of the last $k-1$ positions of the prefix are counted, even if they are not fully contained within the prefix of length $2^k-1$.} of length $2^k-1$ of $S^{(k)}$. Consider the probability spaces $\Omega_k = \F_2^k$ with the uniform measure $\mu_k$, and define the random variables $X_k: \Omega_k \rightarrow \R$ as $X_k(w)= N_w(S^{(k)})$. We say $\cS$ is \emph{$1$-Poisson generic} if 
    \[
    X_k\xrightarrow[k\rightarrow \infty]{d} \cP(1).
    \]
\end{definition}

We comment that the length of the prefix, $2^k-1$, is somewhat arbitrary, and may be replaced with prefixes of length $2^k(1+o(1))$ without changing the essence of Poisson randomness.

For an intuition, we can think of this definition in terms of allocations of balls in bins. Given $k\in \N$ and a binary sequence of length $2^k$, we may think of each window of length $k$ in the sequence as a ball, and each of the possible length-$k$ binary patterns as a bin. Placing the ball corresponding to the $i$-th window inside a bin labeled $w\in\F_2^k$ means that the sequence has the pattern $w$ at position $i$. If we ignore the dependencies between balls corresponding to overlapping windows, distributing $2^k$ balls in $2^k$ boxes uniformly at random should follow a Poisson occupancy law.  

PN sequences satisfy all three of Golomb's randomness postulates. Thus, one might be tempted to think that choosing $\cS=(S^{(k)})_{k\in\N}$ with $S^{(k)}$ being an order-$k$ PN sequence, provides a $1$-Poisson generic family. However, this idea fails completely: in a PN sequence of order $k$, every non-zero window of length $k$ appears exactly once, and thus the random variables $X_k$ converge in distribution to the constant $1$. (This is analogous to the fact that infinite de Bruijn sequences are not Poisson generic). Given the heuristic in the previous paragraph, this does not come as much of a surprise, since a PN sequence corresponds to an allocation in which all but one of the bins contain exactly one ball each.

As a remark, it is possible to give alternative definitions of the Poisson randomness postulate that apply to individual sequences rather than families of sequences, although we will not consider them in this work. For example, one may conceive a Poisson $\varepsilon$-generic sequence to be a periodic sequence, say of length $2^k-1$, such that $d_{TV}(X_k,\cP(1))<\varepsilon$.

\subsection{A candidate family of sequences}
We will now present two families of sequences which empirically seem to satisfy the Poisson convergence condition. For each $k$, let $\alpha_k\in \F_{2^k}^*$ be a primitive element. For ease of notation, we will omit the subscript $k$ in $\alpha_k$ and leave the dependency implicit. Define the two LFSR sequences $S_1^{(k)}, S_2^{(k)}$ as 
\begin{align*}
S_1^{(k)}&=(\Tr(\alpha^i))_{0\leq i\leq 2^k-2}, \\
S_2^{(k)}&=(\Tr(\alpha^{-i}))_{0\leq i\leq 2^k-2}.
\end{align*}
Notice that $S_2^{(k)}$ is the reverse of $S_1^{(k)}$. When $k\geq 3$, $\alpha, \alpha^{-1}$ are both primitive, so it follows that both sequences are PN sequences of minimal period $2^k-1$.

For each $c\in\F_{2^k}$, denote by $S_c^{(k)}$ the LFSR sequence
\[
S_c^{(k)} = (\Tr(\alpha^i+c\cdot \alpha^{-i}))_{0\leq i\leq 2^k-2},
\]
which is the bitwise sum of $S_1^{(k)}$ and a shift of $S_2^{(k)}$. The sequences $S_c^{(k)}$ are codewords in the code $\Melas(k)^\perp$.

Technically, Definition~\ref{def:poisson} calls for infinite sequences, but since we are only interested in the $(2^k-1)$-prefix, we may complete $S_c^{(k)}$ in any way to an infinite sequence, say, periodically. Computing the distribution of the random variables $X_k$ for the sequences $S_c^{(k)}$ as defined in the previous section, we observe that they indeed seem to converge to a $\cP(1)$ distribution. Table \ref{tab:empiricaldistr} shows the point mass probabilities of the random variables $X_k$ for two LFSR sequences, one in $\Melas(15)^\perp$ and one in $\BCH(2,15)^\perp$, and compares them to the $\cP(1)$ distribution by means of the total variation distance $d_{TV}$. The values are rounded to the fourth most significant digit.

\begin{conjecture}\label{conj:melas}
    For any choice of constants $c_k\in \F_{2^k}^*$, the family $\set{S_{c_k}^{(k)}}$ is Poisson generic. 
\end{conjecture}

We also conjecture that this holds for the codewords in the $\BCH(2,k)^\perp$ code:

\begin{conjecture}\label{conj:bch}
    For any choice of constants $c_k\in \F_{2^k}^*$, the family $\set{(\Tr(\alpha^i+c_k\cdot \alpha^{3i}))_{0\leq i\leq 2^k-2}}$ is Poisson generic.
\end{conjecture}

\begin{table}
\caption{Empirical distributions $\mathbb{P}(X_{15}=i)$ for $\Melas(15)^\perp$ and $\BCH(2,15)^\perp$, $i\leq 8$}
\label{tab:empiricaldistr}
\begin{center}
\setlength{\tabcolsep}{4pt}
\begin{tabular}{c|ccccccccc|c}
\hline\hline
    $i$ & 0 & 1 & 2 & 3 & 4 & 5 & 6 & 7 & 8 & $d_{TV}(X_{15}, \cP(1))$\\
    \hline 
    $\Melas(15)^{\perp}$ & 0.3681 & 0.3643 & 0.1897 &  0.06042 & 0.01379 &  2.991 $\times 10^{-3}$&  5.798 $\times 10^{-4}$& 6.104 $\times 10^{-5}$ &  6.104$\times 10^{-5}$ & 0.006132 \\
    $\BCH(2,15)^\perp$ & 0.3655 & 0.3727 & 0.1819 & 0.06107 & 0.01514 & 3.143$\times 10^{-3}$ & 4.272$\times 10^{-4}$ & 6.104 $\times 10^{-5}$ & 3.052$\times 10^{-5}$ & 0.004991\\
    \hline
    $\cP(1)$ & 0.3679 & 0.3679 & 0.1839 & 0.06131 & 0.01533 & 3.066$\times 10^{-3}$ & 5.109$\times 10^{-4}$ & 7.299$\times 10^{-5}$ & 9.124 $\times 10^{-6}$\\
    \hline\hline
\end{tabular}  

\vspace{0.2cm}

 $\Melas(15)^{\perp}$ connection polynomial: $x^{30} + x^{28} + x^{26} + x^{25} + x^{20} + x^{19} + x^{18} + x^{16} + x^{15} + x^{14} + x^{12} + x^{11} + x^{10} + x^5 + x^4 + x^2 + 1$

\vspace{0.2cm}

$\BCH(2,15)^\perp$ connection: $x^{30} + x^{25} + x^{23} + x^{22} + x^{20} + x^{19} + x^{17} + x^{15} + x^{14} + x^{13} + x^{12} + x^{11} + x^9 + x^8 + x^7 + x^5 + x^4 + x^2 + 1$

\vspace{0.2cm}
Initial condition: $[0,0,\dots,0,1]$
\end{center}
\end{table}

We now give a heuristic in an attempt to explain intuitively why these sequences are somewhat natural candidates. Choose two PN sequences of length $2^k-1$, $S_1$ and $S_2$, independently at random and compute their bitwise addition $S_1\oplus S_2$. If we disregard again the dependencies between overlapping windows, we may think of each PN sequence as a random permutation of the length-$k$ windows $00\dots01, 00\dots10,\dots, 11\dots11$. Fix a pattern $w\in\F_2^k$. For each $u\in\F_2^k$, there is exactly one $v=u\oplus w\in\F_2^k\setminus\set{0^k}$ such that $u\oplus v=w$ (this does not apply to $u=w$, but we can ignore it for the sake of the heuristic). The probability that $w$ does \emph{not} appear in $S_1\oplus S_2$ is equal to the probability that none of the pairs $(u,u\oplus w)$ are aligned. This is equivalent to computing the probability that $S_2$ is (as a sequence of $k$-length windows) a permutation of $S_1$ with no fixed point (a derangement). This is known to be
\[
\sum_{i=0}^{2^k-1}\frac{(-1)^i}{i!}\longrightarrow \frac{1}{e},
\]
precisely the probability predicted by $\cP(1)$. More generally, the probability that $w$ appears exactly $j$ times is given by
\[
\frac{\binom{2^k-1}{j}}{(2^k-1)!}(2^k-1-j)!\sum_{i=0}^{2^k-1-j}\frac{(-1)^i}{i!}=\frac{1}{j!}\sum_{i=0}^{2^k-1-j}\frac{(-1)^i}{i!},
\]
which converges to $e^{-1}/j!$ for each fixed $j$.

We point out that the heuristic is still not completely satisfactory: what does it mean for two PN sequences to be ``independent''? Is the reverse of a PN sequence ``independent'' from the original one, which is what Conjecture \ref{conj:melas} might seem to suggest? Also, when $k$ is even, the sequence $(\Tr(\alpha^{3i}))_{0\leq i\leq 2^k-2}$ is not PN, since $\alpha^3$ is not primitive.

\subsection{Study of $\Melas(k)^\perp$}
Let $T_1^{(k)},T_2^{(k)}:\F_{2^k}\rightarrow \F_2^k$ be the $\F_2$-linear transformations given by
\begin{align*}
    T_1^{(k)}(x) & =(\Tr(\alpha^j \cdot x ))_{0\leq j\leq k-1},\\
    T_2^{(k)}(x) & =(\Tr(\alpha^{-j} \cdot x ))_{0\leq j\leq k-1}.
\end{align*}
For convenience of notation, from now on we will avoid the superscript $k$ unless we want to make the dependency explicit. Notice that $T_1(\alpha^i)$ and $T_2(\alpha^{-i})$ are the $i$-th windows in $S_1$ and $S_2$ respectively. This implies that $T_1,T_2$ are both isomorphisms: since $S_1$ and $S_2$ are PN sequences, each non-zero window appears exactly once in each period. Let $L$ be the $\F_2$-linear map defined as $L=T_2^{-1}T_1$. 

Define the probability spaces $\Omega_{k}=\F_2^k$ together with the uniform measure $\mu_k$. Let $N_w(S_c^{(k)})$ be the number of occurrences of $w$ in the first period of $S_c^{(k)}$ (its length $2^k-1$ prefix). Define in each $\Omega_{k}$ the random variables $X_{k,c}:\Omega_k\rightarrow \N_0$ as,
\[
X_{k,c}(w)=N_w(S_c^{(k)}).
\]
Conjecture \ref{conj:melas} says that for every choice of constants $c_1,c_2,\dots$, with $c_k\in\F_{2^k}^*$ for every $k$, 
    \[
X_{k,c_k}\xrightarrow[k\rightarrow \infty]{d} \cP(1).
    \]
For convenience, we will sometimes omit the dependency of $c$ on $k$. To prove the conjecture, it would suffice to show that the factorial moments of $X_{k,c}$ converge to $1$ (see for example \cite[Theorem 1.22]{Bol01}). 

Define the falling factorial $(x)_t\eqdef x(x-1)\dots(x-t+1)$. Let $M_t^{(k,c)} = \E[(X_{k,c})_t]=\E[X_{k,c}(X_{k,c}-1)\dots (X_{k,c}-t+1)]$ be the $t$-th factorial  moment of $X_{k,c}$. Notice that the $i$-th window of $S_c$ can be computed as 
\[
T_1(\alpha^i)+T_2(c\cdot \alpha^{-i})= T_2(T_2^{-1}T_1(\alpha^i)+c\cdot \alpha^{-i}).
\]
Since $T_2$ is an isomorphism, studying the pattern frequencies in $S_c$ is equivalent to studying the sizes of the fibers of the map 
\begin{equation}\label{eq:linearmapequivalence}
    x\mapsto L(x)+cx^{-1}.
\end{equation}
Thus,
\begin{align}
         M_t^{(k,c)} &= \sum_{j\geq 0} (j)_t \cdot \mu(\set{w: N_w(S_c)=j}) \nonumber \\
         &= \frac{1}{2^k}\sum_{w\in\F_{2}^k} (N_w(S_c))_t \nonumber\\
         &= \frac{1}{2^k} \sum_{u\in\F_{2^k}} \#\set{(x_1,\dots,x_t)\in(\F_{2^k}^*)^t: \text{the $x_i$ are distinct, } L(x_1)+cx_1^{-1}=\dots =L(x_t)+cx_t^{-1}=u} \label{eq:momentformula}
    \end{align}

\begin{theorem}
    The following convergence results hold:
    \begin{itemize}
        \item[1)] For every $(c_1,c_2,\dots)\in\prod_{k\in\N}\F_{2^k}^*$, we have $M_1^{(k,c_k)}\xrightarrow[k\rightarrow \infty]{}1$.
        \item[2)] For every $(c_1,c_2,\dots)\in\prod_{k\in\N}\F_{2^k}^*$, we have  $\limsup\limits_{k\rightarrow\infty}M_2^{(k,c_k)}\leq 2$.
        \item[3)] For almost every choice of $(c_1,c_2,\dots)\in\prod_{k\in\N}\F_{2^k}^*$ (with respect to the product of uniform measures), we have $M_2^{(k,c_k)}\xrightarrow[k\rightarrow \infty]{}1$.
    \end{itemize}
\end{theorem}

\begin{IEEEproof}
    1) We have 
    \[
    \E[X_{k,c}] = \sum_{i\geq 0} i\cdot \mu_k(\set{w: N_w( S_c^{(k)})=i}) = \frac{1}{2^k}\sum_{w\in \F_2^k} N_w(S_c^{(k)})=1-\frac{1}{2^k}\longrightarrow 1.
    \]

    2) By \eqref{eq:momentformula}, 
    \[
    M_2^{(k,c)}=\frac{1}{2^k}\#\set*{(x,y)\in (\F_{2^k}^*)^2,  x\neq y: L(x)+cx^{-1}=L(y)+cy^{-1}}. 
    \]
    Given $x\neq y$ in $\F_{2^k}^*$, let $a=x+y\neq 0$, $b=xy\neq 0$. Notice that this means that $x^2+ax+b=0$. The condition defining the set in the above equation is equivalent to
    \[
    b = \frac{c\cdot a}{L(a)}.
    \]
    Thus, for each value of $a$, $b$ is completely determined, which implies that $x$ can take zero or two values, depending on whether the quadratic equation
    \[
    x^2+ax+\frac{ca}{L(a)}=0
    \]
    has roots in $\F_{2^k}$ or not. Solving this equation is equivalent to solving
    \[
        (x/a)^2+(x/a)+\frac{c}{aL(a)}=0.
    \]
    It is known \cite[Theorem 2.25]{LidNie97} that this equation has two solutions if and only if $\Tr\parenv{\frac{c}{aL(a)}}=0$, and no solutions otherwise.
    Thus,
    \begin{equation}\label{eq:secondmoment}
         M_2^{(k,c)} = \frac{1}{2^k} \sum_{a\in\F_{2^k}^*} \parenv*{1+\chi\parenv*{\frac{c}{aL(a)}}}=1-\frac{1}{2^k}+\frac{1}{2^k}\sum_{a\in\F_{2^k}^*}\chi\parenv*{\frac{c}{aL(a)}}\leq 2-\frac{1}{2^{k-1}}.
    \end{equation}

Notice that in the last sum above, the argument of the additive character is a rational function of $a$, since $L$ is a linearized polynomial. Traditional Weil bounds are not relevant in this case due to the high degree of the polynomial. There are some known extensions of the bound for polynomials of low weight-degree, such as \cite{KauLov11}, but they are not useful in the previous computation.

3) First, let us compute the expected value of $M_2^{(k,c)}$ when randomizing over $c$:
\begin{align*}
    \E_c[M_2^{(k,c)}]=\frac{1}{2^{k}(2^k-1)}\sum_{c\in \F_{2^k}^*}\sum_{a\in\F_{2^k}^*} \parenv*{1+\chi\parenv*{\frac{c}{aL(a)}}} &=\frac{1}{2^{k}(2^k-1)} \sum_{a\in\F_{2^k}^*}\parenv*{2^k-1+\sum_{c\in\F_{2^k}^*}\chi\parenv*{\frac{c}{aL(a)}}}\\
    &=\frac{(2^k-2)(2^k-1)}{2^{k}(2^k-1)}=1-\frac{1}{2^{k-1}}.
\end{align*}
Now, notice that for every $u\in\F_{2^k}^*$
\[
\#\set*{a\in \F_{2^k}^*: aL(a)=u}=\#\set*{a\in \F_{2^k}^*: L(a)+ua^{-1}=0}=N_{0^k}(S_u),
\]
so the sum in \eqref{eq:secondmoment} can be rewritten as
\[
M_2^{(k,c)}=1-\frac{1}{2^k}+\frac{1}{2^k}\sum_{u\in\F_{2^k}^*}N_{0^k}(S_u)\chi(cu^{-1}).
\]

Let $W(c)=\sum_{u\in\F_{2^k}^*}N_{0^k}(S_u)\chi(cu^{-1})$. For convenience, let $N_u=N_{0^k}(S_u)$. Regard $W(c)$ and $M_2^{(k,c)}$ as random variables over $\F_{2^k}^*$ considered with the uniform measure. Then
\begin{align*}
\E_{c}[W(c)]&=\frac{1}{2^k-1}\sum_{c\in\F_{2^k}^*}\sum_{u\in\F_{2^k}^*}N_u\chi(cu^{-1})=-\frac{1}{2^k-1}\sum_{u\in\F_{2^k}^*}N_u=-1,\\
    \E_{c}[W(c)^2] &=\frac{1}{2^k-1}\sum_{c\in\F_{2^k}^*}\sum_{u\in\F_{2^k}^*}\sum_{v\in\F_{2^k}^*}N_uN_v\chi(c(u^{-1}+v^{-1}))\\
    &=\sum_{u\in\F_{2^k}^*}N_u^2-\frac{1}{2^k-1}\sum_{u,v\in\F_{2^k}^*, u\neq v}N_uN_v\\
    &=\sum_{u\in\F_{2^k}^*}N_u^2-\frac{1}{2^k-1}\parenv*{\parenv*{\sum_{u\in\F_{2^k}^*}N_u}^2-\sum_{u\in\F_{2^k}^*}N_u^2}\\
    &=\frac{2^k}{2^k-1}\sum_{u\in\F_{2^k}^*}N_u^2-(2^k-1),
\end{align*}
where we have used the orthogonality of characters.

By the correspondence between occurrences of $0^k$ in the LFSR sequences $S_u$, and polynomials of degree $<k$ in polynomial orbits, we see that
\[
\sum_{u\in\F_{2^k}^*}N_u^2= 
    \#\set*{(f,g)\in (\F_2[x]-\set{0})^2 : \deg(f),\deg(g)<k, g(\alpha)g(\alpha^{-1})=f(\alpha)f(\alpha^{-1})} = (k+1)2^{k-1}-1,
\]
by Lemma \ref{lemma:thirdmomentcounting}, Claim 2). Now, a standard computation shows that
\[
\V_{c}[M_2^{(k,c)}]=\frac{1}{2^{2k}}(\E[W(c)^2]-\E[W(c)]^2)=\frac{k-1}{2(2^k-1)}.
\]
By Chebyshev's inequality,
\[
\mathbb{P}_c\parenv*{\abs*{M_2^{(k,c)}-\parenv*{1-\frac{1}{2^{k-1}}}}>\frac{1}{k}}\leq \frac{k^2(k-1)}{2(2^k-1)}.
\]
The result follows from the Borel-Cantelli lemma, since $\sum_k k^3/2^k<\infty$.
\end{IEEEproof}

We will see now that if we average over $c$ (that is, we average the frequency counts over all distinct rotationally non-equivalent LFSRs with connection polynomial $M_\alpha M_{\alpha^{-1}}$), then we can prove convergence of the third moment. 

Define the probability spaces $\widetilde{\Omega}_{k}=\F_{2^k}^*\times\F_2^k$ together with the uniform measure $\widetilde{\mu}_k$. Define in each $\widetilde{\Omega_{k}}$ the random variables $X_{k}:\widetilde{\Omega}_k\rightarrow \N_0$ as:
\[
X_{k}(c,w)= N_w(S_c^{(k)}).
\]
Let $M_t^{(k)} = \E[(X_{k})_t]=\E[X_{k}(X_{k}-1)\dots (X_{k}-t+1)]$ be the $t$-th factorial  moment of $X_{k}$.

\begin{theorem}\label{thm:averagedmoments}
    The following factorial moment converges:
    \[
    M_3^{(k)}\xrightarrow[k\rightarrow \infty]{}1.
    \]

\end{theorem}

\begin{IEEEproof}[Proof of Theorem \ref{thm:averagedmoments}]

    By \eqref{eq:momentformula} adapted to $\widetilde{\Omega}_k$, 
    \[
    M_3^{(k)} = \frac{1}{2^{k}(2^k-1)}\sum_{c\in\F_{2^k}^*}\#\set*{(x,y,z)\in (\F_{2^k}^*)^3:  \text{ $x,y,z$ are distinct, } L(x)+cx^{-1}=L(y)+cy^{-1}= L(z)+cz^{-1}}.
    \]
    The two equalities $x,y,z$ must satisfy are equivalent to 
    \[
    c = xy\frac{L(x+y)}{x+y}=xz\frac{L(x+z)}{x+z}.
    \]
    Let $a=x+y\neq 0$, $b=x+z\neq 0$. Now,
    \[
    (x+a) \frac{L(a)}{a}= (x+b)\frac{L(b)}{b}, 
    \]
    and thus
    \begin{equation}\label{eq:thirdmoment}
        x \parenv*{\frac{L(a)}{a}+\frac{L(b)}{b}}= L(a)+L(b).
    \end{equation}
The right-hand side must be non-zero since $y\neq z$, and thus we conclude there is a solution for $x$ if and only if $\frac{L(a)}{a}+\frac{L(b)}{b}\neq 0$. 
    
    
    Notice that $(c,x,y,z)$ is uniquely determined by the triplet $(a,b,x)$. In order to guarantee $x,y,z\neq 0$ and that they are all distinct, $a,b,x$ should satisfy $a,b\neq 0$, $a\neq b$, $x\neq 0$, $x\neq a$ and $x\neq b$. It is easy to verify using \eqref{eq:thirdmoment} that the last three conditions are redundant. Thus, the third moment is equal to
    \begin{align*}
        M_3^{(k)} & = \frac{1}{2^{k}(2^k-1)}\#\set*{(a,b)\in (\F_{2^k}^*)^2, a\neq b :  \frac{L(a)}{a}\neq \frac{L(b)}{b}}\\
        & = \frac{2^k-2}{2^k}-\frac{1}{2^{k}(2^k-1)}\#\set*{(a,b)\in (\F_{2^k}^*)^2, a\neq b :  \frac{L(a)}{a}= \frac{L(b)}{b}}.
    \end{align*}




    Denote $A  = \#\set{(a,b)\in (\F_{2^k}^*)^2, a\neq b :  \frac{L(a)}{a}= \frac{L(b)}{b}}$.  Given $v\in\F_{2^k}^*$, define the $\F_2$-linear subspaces
    \begin{align*}
        E_v & = \set*{a \in \F_{2^k}: L(a)=v\cdot a},\\
        R_v & = \set*{g\in\F_2[x]: \deg (g)<k, g(\alpha)=v\cdot g(\alpha^{-1})}.
    \end{align*}
    We claim that $\dim E_v = \dim R_v$. On the one hand, $a\in E_v$ if and only if $\Tr(a\cdot(\alpha^j+v\alpha^{-j}))=0$ for every $0\leq j\leq k-1$, which happens if and only if $a\in \langle 1+v,  \alpha+v\alpha^{-1},\dots, \alpha^{k-1}+v\alpha^{-(k-1)}\rangle^{\perp}$. Thus, $\dim E_v= k-\dim \langle\alpha^j+v\alpha^{-j}\rangle_{0\leq j\leq k-1}$. On the other hand, $R_v$ is the kernel of the linear map $\sum_{j=0}^{k-1}g_j x^j\mapsto \sum_{j=0}^{k-1}g_j\cdot (\alpha^j+v\alpha^{-j})$, which is a surjection onto $\langle\alpha^j+v\alpha^{-j}\rangle_{0\leq j\leq k-1}$, and thus has dimension $\dim R_v= k-\dim \langle\alpha^j+v\alpha^{-j}\rangle_{0\leq j\leq k-1}=\dim E_v$.

    We deduce that 
    \begin{align*}
         A = \sum_{v\neq 0} (\#E_v-1)(\#E_v-2) & = \sum_{v\neq 0} (\#R_v-1)(\#R_v-2)  = \sum_{v\neq 0}(\#R_v-1)^2-\sum_{v\neq 0}(\#R_v-1) \\
         & = \#\set*{(f,g)\in (\F_2[x]-\set{0})^2, \deg f,g<k : \frac{g(\alpha)}{g(\alpha^{-1})}=\frac{f(\alpha)}{f(\alpha^{-1})}}-(2^k-1).
    \end{align*}
By Lemma \ref{lemma:thirdmomentcounting}, if $k$ is odd
\[
M_3^{(k)} = 1-\frac{1}{2^{k-1}}-\frac{(3k-3)2^{k-2}+1}{2^k(2^k-1)}+\frac{1}{2^k}= 1-\cO\parenv*{\frac{k}{2^k}}\rightarrow 1,
\]
and for $k$ even we obtain an analogous result, which concludes the proof.  

\end{IEEEproof}

\subsection{Connection to Section \ref{sec:patternfreqandweight}}\label{sec:connection}

 Let $N_{\max}(S_c^{(k)})$ denote the maximum number of occurrences of any pattern $w\in\F_2^k$ in $S_c^{(k)}$. If it were true that all moments of the variables $X_{k,c}$ converged to those of $\cP(1)$, this would imply that they all remain bounded as $k\rightarrow\infty$. In particular, we should have that for every $t\geq 1$, the following remains bounded:
\[
\frac{(N_{\max}(S_c^{(k)}))^t}{2^k}.
\]
In particular, the function $2^{\frac{k}{\log k}}$ satisfies this requirement, which is coherent with Theorem \ref{thm:occurofzero} and Conjecture \ref{conj:klogk}.

So far, our analysis has fallen into one of the following two categories: we either fixed an LFSR sequence and randomized over the patterns, or we randomized over both sequences and patterns. In both cases, we observe a Poisson behavior. But what if we fix a pattern and randomize over the sequences instead? More precisely, given a fixed $w_k\in \F_2^k$, define the random variables $X_{w_k}$ over $\F_{2^k}^*$ with the uniform measure as $X_{w_k}(c)=N_{w_k}(S_c)$. Is it true that for every choice of $w_1,w_2,\dots$, we have $X_{w_k}\xrightarrow[k\rightarrow \infty]{d} \cP(1)$? The answer is no! For $w=0^k$, Theorem \ref{thm:notfullsymbolweight} with $s=0$ implies that $X_{0^k}\xrightarrow[k\rightarrow \infty]{d} 0$. 

Nonetheless, we conjecture that for almost every choice of $(w_1,w_2,\dots,w_k,\dots)\in \prod_{k\geq 1}\F_2^k$ with respect to the product of uniform measures, the distribution does indeed approach the $\cP(1)$ distribution. Table \ref{tab:fixedw} shows the empirical point mass distributions rounded to $4$ significant digits, for $k=8$ and six possible values of $w$. We can observe that indeed for $w=0^8$ (and $w=0^71$ as well), the distributions are far away from $\cP(1)$, while for the remaining four patterns the distance is much smaller.

\begin{table}
\caption{Empirical distribution $\mathbb{P}(X_w=i)$ for $\Melas(8)^\perp$, $w$ fixed, $i\leq 3$}
\label{tab:fixedw}
\begin{center}
\begin{tabular}{c|cccc|c}
\hline\hline
    $i$ & 0 & 1 & 2 & 3 & $d_{TV}(X_w, \cP(1))$\\
    \hline 
    $w=00000000$ & 0.702 & 0.02745 & 0.1255 & 0.01961 & 0.4406\\
   
    $w=00000001$ & 0.5176 & 0.102 & 0.2863 & 0.06275 & 0.2665 \\
    
    $w=00000010$ & 0.3686 & 0.3569 & 0.2 & 0.06275 & 0.0191\\

    $w=00000011$ & 0.3647 & 0.3686 & 0.1843 & 0.07059 & 0.0104\\

    $w=00000100$ & 0.3608 & 0.3569 & 0.2196 & 0.05098 & 0.03567\\

    $w=00000101$ & 0.3412 & 0.3882 & 0.2157 & 0.04706 & 0.05296\\

    \hline
    $\cP(1)$ & 0.3679 & 0.3679 & 0.1839 & 0.06131 & \\
    \hline\hline
\end{tabular}
\vspace{0.2cm}

Connection polynomial: $x^{16} + x^{14} + x^{13} + x^{11} + x^8 + x^5 + x^3 + x^2 + 1$
\end{center}
\end{table}



\subsection{Another relaxation}
Given our previous observation that studying the pattern frequencies in $S_c$ is equivalent to studying fibers of the map \eqref{eq:linearmapequivalence}, a natural relaxation of the problem would be to randomize over all possible $\F_2$-linear transformations $L\in\cL(\F_{2^k})$. We could also randomize over all invertible linear transformations, although the arguments become a bit more involved. 

\begin{theorem}\label{thm:poissonlinearrelaxation}
    For each $k\geq 1$ consider the space $\Omega'_k= \F_{2^k} \times \cL(\F_{2^k})$ together with the uniform measure $\mu'_k$. Define the random variables $X_k': \Omega'_k \rightarrow \N_0$ as:
    \[
    X_k'(u,L) = \#{\set*{x\in\F_{2^k}^*: x^{-1}+L(x)=u}}.
    \]
    Then,
    \[
    X_k' \xrightarrow[k\rightarrow \infty]{d} \cP(1).
    \]
\end{theorem}

\begin{IEEEproof}
    The $t$-th factorial moment is equal to
    \[
         \frac{1}{2^k\cdot 2^{k^2}} \sum_{u,L} \#\set{0\leq j_1,\dots, j_t\leq 2^k-2 \text{ distinct}: \alpha^{-j_1}+L(\alpha^{j_1})=\dots =\alpha^{-j_t}+L(\alpha^{j_t})=u}.
    \]
   Let us separate the counting into different cases.
   
   \textbf{Case 1}: $\set{\alpha^{j_1},\dots,\alpha^{j_t}}$ is linearly independent. There are $\prod_{i=0}^{t-1}(2^k-2^i)$ choices of $j_1,\dots,j_t$ that correspond to this case. For each such choice and for each value of $u\in\F_{2^k}$, the equations $L(\alpha^{j_i})=u+\alpha^{-j_i}$, $1\leq i\leq t$, can be extended in $2^{k(k-t)}$ ways to define a linear transformation $L\in \cL(\F_{2^k})$. Thus, the contribution of these elements to the sum is:
   \[
   \frac{1}{2^k\cdot 2^{k^2}}\cdot \prod_{i=0}^{t-1}(2^k-2^i)\cdot 2^k \cdot 2^{k(k-t)} =\frac{1}{2^{kt}} \prod_{i=0}^{t-1}(2^k-2^i) \xrightarrow[k\rightarrow\infty]{} 1,
   \]
   and the convergence holds for each fixed $t$.

   \textbf{Case 2}: $\set{\alpha^{j_1},\dots,\alpha^{j_t}}$ is linearly dependent. Let $d = \text{dim}(\langle\alpha^{j_1},\dots,\alpha^{j_t}\rangle) <t$ and let $s$ be the size of some linearly dependent subset $\set{\alpha^{j_{\sigma_1}},\dots,\alpha^{j_{\sigma_s}}}$ satisfying $\sum_{i=1}^s  \alpha^{j_{\sigma_i}}=0$. Then $L$ has to satisfy $\sum_{i=1}^s L(\alpha^{j_{\sigma_i}})=\sum_{i=1}^s \alpha^{-j_{\sigma_i}}+s\cdot u=0$. We separate into two sub-cases:

   \textbf{Case 2.a.}: $s$ is odd for some linearly dependent subset. Then, $u=\sum_{i=1}^s \alpha^{-j_{\sigma_i}}$ and there is at most one possible value of $u\in\F_{2^k}$ given a choice for $j_1,\dots,j_t$, and $L$ can be extended in $2^{k(k-d)}$ ways. By the dimension assumption, the number of choices for $j_1,\dots,j_t$ is at most $t!\prod_{i=0}^{d-1}(2^k-2^i)\prod_{i=0}^{t-d-1} (2^d-d-i)\leq t! 2^{kd} 2^{d(t-d)}$. Hence, this case contributes at most:
   \[
   t!\sum_{d=1}^{t-1}\frac{1}{2^k\cdot 2^{k^2}} 2^{kd+d(t-d)}2^{k(k-d)}=\frac{t!}{2^k}\sum_{d=1}^{t-1}2^{d(t-d)}\xrightarrow[k\rightarrow\infty]{}0,
   \]
   and again the convergence holds for each fixed $t$.

   \textbf{Case 2.b.}: For every possible choice of $j_1,\dots,j_s$ such that $\sum_{i=1}^s  \alpha^{j_{\sigma_i}}=0$, $s$ is even. In this case, $\sum_{i=1}^s \alpha^{-j_{\sigma_i}}=0$ and there is possibly no restriction on $u$. Choose $\set{j_{\sigma_1},\dots,j_{\sigma_s}}$ in particular as a minimal linearly dependent set. Note that $d\geq s-1$.

   Let $\alpha^{j_{\sigma_i}}=x_i$ for $1\leq i\leq s$. We deduce that $x_1^{-1}+\dots +x_{s-1}^{-1}+(x_1+\dots+x_{s-1})^{-1}=0$, which implies that
   \[
   \sum_{\substack{1\leq i,j \leq s-1\\i\neq j}}x_i^2\prod_{\substack{1\leq \ell \leq s-1\\\ell \neq i,j}}x_{\ell}=0. 
   \]
   The above polynomial has total degree $s-1$. By the Schwartz-Zippel lemma~\cite{Zip79,Sch80}, the number of solutions $(x_1,\dots,x_{s-1})\in \F_{2^k}^{s-1}$ is at most $(s-1)2^{k(s-2)}$.



   Once $\set{j_{\sigma_1},\dots,j_{\sigma_s}}$ is chosen, we need to choose $d-(s-1)=d-s+1$ additional linearly independent elements, of which there are at most $(2^k)^{d-s+1}$ options, and we need to choose the remaining $t-s-(d-s+1)=t-d-1$ elements, which have to belong to the $d$-dimensional subspace of the previously chosen elements, and thus there are at most $(2^d)^{t-d-1}$ options. As we mentioned before, $u$ could take $2^k$ values. 

   Now, there are $2^{k(k-d)}$ possible extensions for $L$. All in all, we get the following upper bound on the contribution for this case:
   \begin{align*}
&\frac{1}{2^k2^{k^2}}\sum_{\substack{(s,d)\text{ such that} \\s\leq t, t\geq d\geq s-1 \\s \text{ even}}}(s-1)2^{k(s-2)} 2^{k(d-s+1)}2^{d(t-d-1)}2^{k(k-d)}2^k\\
& \qquad \leq \frac{1}{2^{k^2}}\sum_{\substack{(s,d)\text{ such that} \\s\leq t, t\geq d\geq s-1 \\s \text{ even}}} \frac{s-1}{2^{2k}}2^{ks}2^{k(d-s+1)}2^{d(t-d-1)}2^{k(k-d)}\\
&\qquad \leq \frac{1}{2^{k}}\sum_{\substack{(s,d)\text{ such that} \\s\leq t, t\geq d\geq s-1 \\s \text{ even}}}(s-1) 2^{d(t-d-1)}
  \end{align*}
   which converges to $0$ for each fixed $t$.
\end{IEEEproof}

Notice that Theorem \ref{thm:poissonlinearrelaxation} does not shed light on our original problem. Out of the $2^{k^2}$ choices of linear transformation $L$, only $2^k$ of these actually correspond to a binary LFSR sequence. However, it does give information on the average  distribution of values of LFSR sequences over $\F_{2^k}$ with connection polynomial $M_{\alpha}M_{\alpha^{-1}}$. Since $M_{\alpha}M_{\alpha^{-1}}$ factorizes linearly as $(x+\alpha)(x+\alpha^2)\dots (x+\alpha^{2^{k-1}})(x+\alpha^{-1})\dots (x+\alpha^{-2^{k-1}})$ over $\F_{2^k}$, such a sequence is equal to
\[
(a_0\alpha^i+a_1(\alpha^i)^2+\dots +a_{k-1}(\alpha^i)^{2^{k-1}}+b_0\alpha^{-i}+b_1(\alpha^{-i})^2+\dots+ b_{k-1}(\alpha^{-i})^{2^{k-1}})_{0\leq i\leq 2^k-2}
\]
for some choice of $a_0,\dots,a_{k-1}, b_0, \dots, b_{k-1}$. If we let $L_1(x)= \sum_{i=0}^{k-1}a_i x^{2^i}$ and $L_2(x)= \sum_{i=0}^{k-1}b_i x^{2^i}$, $L_1$ and $L_2$ are linearized polynomials, and thus, linear transformations. The sequence is equal to $(L_1(\alpha^i)+L_2(\alpha^{-i}))_{0\leq i\leq 2^k-2}$, so in the case $L_2$ is invertible, the distribution of values of the sequence is the same as the distribution of values of the map $L_2^{-1}L_1(x)+x^{-1}$.

It is worth noting that if we instead wished to study the distribution of values of the map $f(x)$, where $f$ is a random polynomial over $\F_{2^k}$, there are similar results available in the literature. For example, the number of roots that lie in $\F_{2^k}$ of a polynomial of degree $<2^k$ over $\F_{2^k}$ chosen uniformly at random, converges in distribution to a $\cP(1)$ random variable \cite{Leo06}.

\section{Conclusion}

In order to shed a different light on the results of this work, we now give some concluding remarks addressing its contents in the opposite order to the one they were presented in.

In this work we introduced a new randomness postulate on families of binary sequences, motivated by the empirical observation that codewords in the dual Melas codes exhibit a very clear Poisson behavior of pattern frequencies. Proving that this is indeed the case seems to be an extremely hard problem, and we began by proving the convergence of the first three factorial moments, as well as a relaxed version of the conjecture.

One of the implications of the conjecture that all factorial moments converge to those of the $\cP(1)$ distribution, is that patterns of length $k$ cannot appear too many times in a codeword of the dual Melas code, $\Melas(k)^\perp$. More precisely, for every $t$, the maximum number of occurrences of a pattern $N_{\max}$ should satisfy that $(N_{\max})^t2^{-k}$ remains bounded. To put into perspective how far away previous results were from this, the best known bounds were of the order of $2^{k/2}$, which only satisfy this property for $t\leq 2$. 

By translating the original pattern counting problem in LFSRs into a question of polynomial orbits, we managed to prove that the $0^k$ pattern does indeed satisfy the above property for every $t$: it appears at most (approximately) $2^{k/\log_2 k}$ times in any codeword of $\Melas(k)^\perp$. This is strictly better than any previously available result. We employed a similar approach to study bounds on the number of full $k$-symbol weight codewords in $\Melas(k)^{\perp}$ and to give an improvement on known bounds for the burst-covering radius of Melas codes and the critical exponent of dual Melas codes.

We note that there is an alternative way of thinking about the counting problem in LFSRs. By Equation \eqref{eq:fourier}, the number of occurrences $N_w$ of patterns $w\in\F_2^k$ are the coefficients of the Walsh-Hadamard transform of the function $\varphi: \F_2^k\rightarrow \R$. Bounding the maximum number of occurrences means finding a bound for the infinity norm of the transform, while the Poisson moment estimates deal with the $L^1, L^2, L^3,\dots$ norms.

Many questions remain. We have presented multiple conjectures throughout this work, many of which seem unapproachable by current techniques, and which we hope will motivate fruitful future research.

\ \\
\textbf{AI use disclosure}: The idea of exploiting the connection between LFSR sequences and polynomial orbits in order to improve the known bounds on the $b$-symbol weight was the authors', as well as the idea and intuitions behind the definition and conjectures of the Poisson randomness postulate. Proofs in this paper were conceived with both human and AI input (ChatGPT 6.1 and Claude Opus 5). The paper was written by the authors in its entirety.

\bibliographystyle{IEEEtranS}
\bibliography{allbib}

\appendix

\begin{IEEEproof}[Proof of Lemma \ref{lemma:genfuncidentities}]
    a) $\cH(z)\cU(z)= \prod_{h\in H}(1+z^{\deg h}+z^{2\deg h}+\dots)\prod_{u\in U}(1+z^{\deg u}+z^{2\deg u}+\dots)$. The coefficient of $z^d$ in this product is equal to the number of choices of $((a_h)_{h\in H}, (n_u)_{u\in U})$ such that $\sum a_h \deg(h)+\sum n_u \deg(u)=d$, which is exactly $\sigma(d)$.

    b) Recall that $\set{1,x}\not\in U$. Then,
    \[
    \frac{1}{1-z}\cH(z)\cU^2(z)= \prod_{f \text{ irreducible}}\frac{1}{1-z^{\deg f}}= \sum_{d\geq 0}2^d z^d=\frac{1}{1-2z}.
    \]
    The next-to-last equality holds because by the uniqueness of polynomial factorization, the coefficient of $z^d$ is exactly the number of polynomials of degree $d$.
    
    
    c) Recall that $\tau(d)$ is the number of elements of $H$ of degree $d$. Then
    \[
    \tau(d) \text{ is }\begin{cases}
        =1 & \text{ for } d=1\\
        =0 & \text{ if } 2\nmid d, d>1 \\
        \leq 2^{d/2+1} & \text{ otherwise}
    \end{cases},
    \]
    where the last inequality is due to the elements of $H$ being self-reciprocal. Taking logarithm in the definition of $\cH(z)$,
    \[
    \log \cH(z) = \log\parenv*{\frac{1}{1-z}}+\sum_{d\geq 2}\tau(d) \log\parenv*{\frac{1}{1-z^d}}=\log\parenv*{\frac{1}{1-z}}+\sum_{d\geq 2}\sum_{n\geq 1} \tau(d) \frac{z^{dn}}{n}.
    \]
    Suppose $|z|\leq 1/\sqrt{\delta}$ with $\delta>2$. Then:
    \[
    \sum_{d\geq 2}\sum_{n\geq 1} \tau(d) \frac{|z|^{dn}}{n}\leq \sum_{2|d}\sum_{n\geq 1}\frac{2^{d/2+1}}{\delta^{dn/2}n}= 2\sum_{n\geq 1}\frac{1}{n}\sum_{2|d}\parenv*{\frac{2}{\delta^{n}}}^{d/2}= 2\sum_{n\geq 1}\frac{2}{n(\delta^n-2)}<\infty.
    \]
    and thus we have absolute convergence when $|z|<1/\sqrt{2}$. The convergence is uniform in each disc $|z|\leq 1/\sqrt{\delta}$, $\delta>2$. Furthermore, it can be verified that $z=1/\sqrt{2}$ is a singularity: by \cite{Car67}, for $d$ even we have
    \[\tau(d) =\frac{1}{d}\sum_{j|\frac{d}{2}, j\text{ odd}}\mu(j)2^{\frac{d}{2j}}\geq \frac{2^{d/2}}{d}-2^{d/6},\]
    so for $n=1$ the sum in $d$ is divergent when $z=1/\sqrt{2}$.
\end{IEEEproof}

\begin{IEEEproof}[Proof of Theorem \ref{thm:sigmaasymptotics}]
     First, we have that
    \[
    (1-2z)^{-1/2}= \sum_{n\geq 0}\binom{-1/2}{n}(-2z)^n = \sum_{n\geq 0} \frac{1}{2^n}\binom{2n}{n}z^n.
    \]
    Let $b_n=\frac{1}{2^n}\binom{2n}{n}$. By \cite{Hir15}, we have that
    \[
    b_n = \frac{2^n}{\sqrt{\pi n}}(1-\delta_n),
    \]
    with $0\leq \delta_n\leq \frac{1}{4n}$.
    
    By the proof of Lemma \ref{lemma:genfuncidentities}, Claim c), $\log \cH(z)$ is finite in $|z|<1/\sqrt{2}$, so $(1-z)\cH(z)$ has no zeros in this disc. Thus, $A(z)\eqdef e^{\frac{1}{2}(\log(1-z)+\log\cH(z))}$ is holomorphic in $|z|<1/\sqrt{2}$ and satisfies $A^2(z)=(1-z)\cH(z)$, $A(0)=1$. Since the formal power series $(1-z)\cH(z)$ has a unique formal square root which is $1$ at $z=0$, then the coefficients of the Taylor expansion at $0$ of $A(z)$ are the $a_n$ given by the binomial expansion. 
    
    By Lemma \ref{lemma:genfuncidentities}, Claim c), the smallest modulus of a singularity of $A(z)$ is $1/\sqrt{2}$. Fix some $1/2<r<1/\sqrt{2}$. By Cauchy's estimate, there exists a constant $C_r$ such that
    \[
    |a_n|\leq \frac{C_r}{r^n}.
    \]
    By Theorem \ref{thm:genfuncbigequality}, 
    \[
    \sigma(d)= \sum_{j=0}^d a_j b_{d-j}.
    \]
    We split the analysis of the sum into the two intervals $[0,d/2)$ and $[d/2,d]$.
    
    For $0\leq j< d/2$, we have
    \[
    \sum_{j<d/2}a_j b_{d-j} = \frac{2^d}{\sqrt{\pi d}}\sum_{j<d/2}\frac{a_j}{2^j}\sqrt{\frac{d}{d-j}} (1-\delta_{d-j}).
    \]
    Let $\sqrt{\frac{d}{d-j}}=1+\varepsilon_{j,d}$. Then,
    \begin{equation}\label{eq:jlessdovertwo}
    \begin{split}
        \sum_{j<d/2}a_j b_{d-j} &= \frac{2^d}{\sqrt{\pi d}}\sum_{j<d/2}\frac{a_j}{2^j}(1+\varepsilon_{j,d}) (1-\delta_{d-j}) \\
        &=\frac{2^d}{\sqrt{\pi d}}\parenv*{\sum_{j<d/2}\frac{a_j}{2^j}+\sum_{j<d/2}\frac{\varepsilon_{j,d}a_j}{2^j}-\sum_{j<d/2}\frac{a_j\delta_{d-j}}{2^j}-\sum_{j<d/2}\frac{a_j\varepsilon_{j,d}\delta_{d-j}}{2^j}}
        \end{split}
    \end{equation}
    The first sum in the parenthesis is 
    \[
    \sum_{j<d/2}\frac{a_j}{2^j}= \cA(1/2)-\sum_{j>d/2}\frac{a_j}{2^j},
    \]
    and
    \[
    \abs*{\sum_{j>d/2}\frac{a_j}{2^j}}\leq C_r\sum_{j>d/2}\frac{1}{(2r)^j}= C_r\frac{2r}{2r-1}\frac{1}{(2r)^{d/2}},
    \]
     which for $r>1/2$ is $\cO(1/d)$.  
    From $\sqrt{1+x}\leq 1+\frac{1}{2}x$, we deduce that $\varepsilon_{j,d}\leq \frac{j}{2(d-j)}\leq \frac{j}{d}$. Thus, the second term is upper bounded by
    \[
    \abs*{\sum_{j<d/2}\frac{\varepsilon_{j,d}a_j}{2^j}}\leq \frac{1}{d} \sum_{j=0}^\infty\frac{j|a_j|}{2^j}=\cO(1/d).
    \]
    The third term satisfies 
    \[
    \abs*{\sum_{j<d/2}\frac{a_j\delta_{d-j}}{2^j}}\leq \sum_{j<d/2}\frac{|a_j|}{2^j\cdot 4(d-j)}\leq \frac{1}{2d}\sum_{j=0}^\infty \frac{|a_j|}{2^j}=\cO(1/d).
    \]
Finally, the fourth term is bounded by
    \[
    \abs*{\sum_{j<d/2}\frac{a_j\varepsilon_{j,d}\delta_{d-j}}{2^j}}\leq \frac{1}{2d^2}\sum_{j=0}^\infty\frac{j|a_j|}{2^j}=\cO(1/d^2).
    \]
    
    It only remains to analyze the sum in the range $d/2\leq j\leq d$. Since $b_{d-j}\leq 2^{d-j}$,
    \[
    \abs*{\sum_{j=d/2}^d a_j b_{d-j}}\leq 2^dC_r\sum_{j=d/2}^d\frac{1}{(2r)^j}\leq 2^d C_r \frac{2r}{2r-1}\frac{1}{(2r)^{d/2}}=  \frac{2rC_r}{2r-1} \parenv*{\frac{2}{\sqrt{2r}}}^{d},
    \]
    which for $r>1/2$ is exponentially smaller than $2^d$.
\end{IEEEproof}

\begin{IEEEproof}[Proof of Theorem \ref{thm:sigmasumestimate}]
     By Theorem \ref{thm:sigmaasymptotics}, we have
    \begin{align*}
    \sum_{0\leq d<D}\sigma(d)&=1+\cA(1/2)\sum_{1\leq d< D}\frac{2^{D-d}}{\sqrt{\pi (D-d)}}(1+\cO(1/(D-d)))\\
    &= 1+ \cA(1/2)\frac{2^D}{\sqrt{\pi D}} \sum_{1\leq d<D} \frac{1}{2^d}\sqrt{\frac{D}{D-d}}(1+\cO(1/(D-d))).
    \end{align*}
    An analysis almost identical to that of \eqref{eq:jlessdovertwo} in the proof of Theorem \ref{thm:sigmaasymptotics} reveals that the last sum is $1+\cO(1/D)$ in the range $d<D/2$. For $d\geq D/2$, the sum is bounded by
    \[
    C\cdot\sqrt{D}\sum_{d=D/2}^{\infty}\frac{1}{2^{d}}, 
    \]
    where $C$ is a constant. This value is exponentially smaller than $1/D$.
\end{IEEEproof}

\begin{IEEEproof}[Proof of Lemma \ref{lemma:thirdmomentcounting}]
    1) Let us show there is a bijection 
    \begin{multline*}
     \cM=\set*{\begin{aligned}
         (f,g) \text{ with } f,g\neq 0 \text{ and } \deg (f),\deg (g)<k  \\
    \text{such that } g(\alpha)f(\alpha^{-1})=f(\alpha)g(\alpha^{-1})
     \end{aligned}}\longleftrightarrow \\
     \set*{\begin{aligned}
         (r,h,A,B) \text{ with } r\in \Z, 0\leq r\leq \floor*{\frac{k-1}{2}} \text{ and } h,A,B \text{ polynomials }\\
         \text{ such that }h\neq 0, \deg (h)\leq k-1-2r,  \gcd(A,B)=1, \max(\deg (A), \deg (B))=r
     \end{aligned}
 } =\cN
    \end{multline*}
    
Let $f,g\in \F_2[x]$ be both non-zero and $\deg(f),\deg(g)<k$. Let $h=\gcd(f,g)$, so that $f=hu$, $g=hv$ with $\gcd(u,v)=1$. By Lemma \ref{lemma:orbitequivdef}, Claim 3), the condition $g(\alpha)f(\alpha^{-1})=f(\alpha)g(\alpha^{-1})$ is equivalent to $g(x)f(x^{-1})=f(x)g(x^{-1})$, which in turn is equivalent to $u(x)v(x^{-1})=u(x^{-1})v(x)$. Let $d = \max(\deg (u), \deg (v))$. Multiplying by $x^d$, we obtain
    \[
    (x^d u(x^{-1}))\cdot v(x)= (x^d v(x^{-1}))\cdot u(x).
    \]
    By coprimality of $u$ and $v$, we deduce that there exists a polynomial $w$ such that 
    \[
    \begin{cases}
        x^d u(x^{-1}) & = u(x) w(x), \\
        x^d v(x^{-1}) & = v(x) w(x).
    \end{cases}
    \]
    Since $\gcd(u,v)=1$, at least one of $u,v$ evaluates to $1$ at $0$. Without loss of generality, suppose $u(0)=1$. Then $\deg (x^d u(x^{-1}))= d = \deg (u) +\deg (w)$, and $\deg (v) +\deg (w)=\deg (x^d v(x^{-1}))\leq d$. Since $d$ is equal to $\deg (u)$ or $\deg (v)$, we deduce in either case from the two previous equations that $\deg (w)=0$, so $w=1$. Hence,
    \[
    \begin{cases}
        x^d u(x^{-1}) & = u(x),  \\
        x^d v(x^{-1}) & = v(x),
    \end{cases}
    \]
    which means that $u_i=u_{d-i}$, $v_i=v_{d-i}$ for $0\leq i\leq d$. Notice that if $d$ were odd, then the number of non-zero coefficients in $u$ and $v$ would be even in both cases, so $u(1)=v(1)=0$, a contradiction. Thus, $d=2r$ for some $r\geq 0$. 

    Notice that $\set{x^r, x^r(x+x^{-1}), x^r (x+x^{-1})^2,\dots, x^r(x+x^{-1})^r}$ is a basis of the linear subspace $\set{z\in \F_2[x], \deg z\leq 2r: x^{2r}z(x^{-1})=z(x)}$. To see it, first note that the subspace has dimension $r+1$. All $r+1$ polynomials in the list are linearly independent, as they have degrees $r,r+1,\dots 2r$, and they all belong to the subspace. Thus, there exist $A,B\in\F_2[x]$ of degrees $\leq r$ such that  
    \[
    \begin{cases}
        u(x) & = x^r A(x+x^{-1}),  \\
        v(x) & = x^r B(x+x^{-1}),
    \end{cases}
    \]
    and at least one of $A,B$ has degree exactly $r$ (since at least one of $u,v$ has degree $2r$). As a side note, the operation $f\mapsto x^{\deg f}f(x+x^{-1})$ is known as the Q-transform \cite{Meyn90}.

    We claim that $\gcd(A,B)=1$. Indeed, suppose $\beta$ were a common root of $A$ and $B$ in some extension field. If $\gamma$ is a root of $x^2+\beta x+1=0$, then $\gamma$ would be a common root of $u$ and $v$, which is a contradiction. 

    Thus, we have shown that a valid pair $(f,g)$ defines a choice of $(h,r,A,B)$ with $\deg(h)\leq k-1-2r$, $0\leq r\leq \floor{\frac{k-1}{2}}$, $A,B$ coprime, $\max(\deg (A), \deg (B))=r$. The converse is also true: given a valid choice of $A$ and $B$ it is not hard to see that defining $u=x^r A(x+x^{-1})$ and $v=x^r B(x+x^{-1})$ yields two coprime polynomials. 

    Let $C_r=\#\set{(A,B): A,B\neq 0, \gcd(A,B)=1, \max(\deg (A),\deg (B))=r}$. Then, 
    \[
    \#\cN= \sum_{r=0}^{\floor{\frac{k-1}{2}}} C_r (2^{k-2r}-1),
    \]
    since once $r$, $A$ and $B$ are chosen, $h$ can be any polynomial of degree at most $k-2r-1$. 

    Notice that $C_0=1$. There are many results in the literature on computing the probability that two polynomials of a given degree chosen uniformly at random are coprime. In particular, \cite[Corollary 2]{BenBen07} implies that
    \[
    \sum_{i=0}^{r-1}C_i=2^{2r-1}-1.
    \]
    It is not hard to show by induction that this implies $C_r=3\cdot 2^{2r-1}$ for $r\geq 1$. Plugging this result into the formula for $\#\cN$ we obtain the claimed result.

    2) Let us see how many ways of choosing $f,g$ satisfying $f(\alpha)f(\alpha^{-1})=g(\alpha)g(\alpha^{-1})$ there are. Let $h=\gcd(f,g)$, and $f=hF$, $g=hG$ with $\gcd(F,G)=1$. Only one of $F, G$ can be divisible by $x$ (possibly neither).
    
    Suppose $x\nmid F,G$. Since $f(\alpha)f(\alpha^{-1})=g(\alpha)g(\alpha^{-1})$, we have that $F(\alpha)F(\alpha^{-1})=G(\alpha)G(\alpha^{-1})$, which by Lemma \ref{lemma:orbitequivdef}, Claim 1), is equivalent to $F(x)\overleftarrow{F}(x)=G(x)\overleftarrow{G}(x)$. Since $F$ and $G$ need to be coprime, we must have necessarily that $G=\overleftarrow{F}$ and that $\gcd(F,\overleftarrow{F})=1$. Let us denote by $r_t$ the number of polynomials $P$ of degree $t$ such that $P(0)=1$ and $\gcd(P,\overleftarrow{P})=1$. Then, for each $0\leq d\leq k-1$, the number of valid choices of $F,G$ such that $\deg (F)=d$ is $r_d$. For each choice, there are $2^{k-d}-1$ possible non-zero polynomials $h$ of degree at most $k-d-1$.

    Suppose now that $F=x^aP$, $P(0)=1$. Since $f(\alpha)f(\alpha^{-1})=g(\alpha)g(\alpha^{-1})$, we have that $P(\alpha)P(\alpha^{-1})=G(\alpha)G(\alpha^{-1})$. By the same argument as before, $G=\overleftarrow{P}$ and $\gcd(P,\overleftarrow{P})=1$. Now, the number of options for $F,G$ is given by $\sum_{i=0}^{d-1}r_i$, where the $i$-th term corresponds to choosing $a=d-i$. Again, once $F,G$ have been selected, the number of options for $h$ is $2^{k-d}-1$.
    
    The case when $x|G$ is completely analogous. In summary, the counting is equal to:
    \begin{equation}\label{eq:horriblelemma}
       \sum_{d=0}^{k-1}(2^{k-d}-1)\parenv*{2\sum_{a=0}^{d-1}r_a+r_{d}}, 
    \end{equation}

    Let $s_d=2\sum_{a=0}^{d-1}r_a+r_{d}$. Notice that $s_d-s_{d-1}=r_d+r_{d-1}$. Observe that any polynomial $A(x)$, with $A(0)=1$, is decomposed in a unique way as a product $H(x)B(x)$, where $H=\overleftarrow{H}$, $B(0)=1$ and $\gcd(B,\overleftarrow{B})=1$. Since the number of self-reciprocal polynomials of degree $t$ is $2^{\floor{\frac{t}{2}}}$, we deduce that
    \[
    \sum_{t=0}^d r_{d-t} 2^{\floor*{\frac{t}{2}}}=2^{d-1},
    \]
    for $d\geq 1$. Subtracting twice this equation instantiated in $d-2$, all terms cancel except for the first two, so $r_d+r_{d-1}=2^{d-2}$ for $d\geq 3$. Then, $s_d-s_{d-1}=r_d+r_{d-1}=2^{d-2}$, and using $r_0=1$, $r_1=r_2=0$ we can show by induction that $s_0=1$, $s_1=2$, $s_d=2^{d-1}$ for $d\geq 2$. Substituting this into \eqref{eq:horriblelemma}, we obtain the desired result.
\end{IEEEproof}

\end{document}